\documentclass[11pt,a4paper]{article}
\usepackage{amssymb}
\usepackage{amsmath, amscd}
\usepackage{amsfonts}
\usepackage{graphicx}
\usepackage{a4wide}
\usepackage{microtype}
\usepackage{bbm}
\usepackage{slashed}
\usepackage{times}
\usepackage{amsthm}
\usepackage{mathrsfs}
\usepackage{hyperref}
\usepackage{color}

\definecolor{darkred}{rgb}{0.8,0.1,0.1}
\hypersetup{
     colorlinks=true,         
     linkcolor=darkred,
     citecolor=blue,
}

\usepackage{comment}
\usepackage[margin=2.2cm]{geometry}
\usepackage[all,cmtip]{xy}

\usepackage{marvosym}

\theoremstyle{plain}

\theoremstyle{definition}

\definecolor{NiColor}{RGB}{77,77,255}
\definecolor{NiColoRed}{RGB}{255,77,77}
\definecolor{NiCitation}{RGB}{77,255,77}

\newcommand{\wick}[1]{:\!{#1}\!:}

\newtheoremstyle{TheoremStyle}
{3pt}
{3pt}
{}
{}
{\bf}
{:}
{.5em}
{}
\newtheoremstyle{ExampleAndRemarkStyle}
{3pt}
{3pt}
{}
{}
{\bf}
{:}
{.5em}
{}
\newtheoremstyle{ProofStyle}
{3pt}
{3pt}
{}
{}
{\bf}
{:}
{.5em}
{}

\theoremstyle{TheoremStyle}
\newtheorem{theorem}{Theorem}

\newtheorem{proposition}[theorem]{Proposition}
\newtheorem{lemma}[theorem]{Lemma}
\newtheorem{Definition}[theorem]{Definition}
\theoremstyle{ExampleAndRemarkStyle}
\newtheorem{remark}[theorem]{Remark}
\newtheorem{Example}[theorem]{Example} 
\theoremstyle{ProofStyle}

\title{%
Ricci Flow from the Renormalization of Nonlinear Sigma Models in the Framework of Euclidean Algebraic Quantum Field Theory
}

\author{%
Mauro Carfora$^{1,2,3,a}$, Claudio Dappiaggi$^{1,2,3,b}$, Nicol\`o Drago$^{1,2,3,c}$ 
and Paolo Rinaldi$^{1,d}$\vspace{4mm}\\
{\small $^1$ Dipartimento di Fisica -- Universit{\`a} di Pavia, Via Bassi 6, 27100 Pavia, Italy.}\vspace{2mm}\\
{\small $^2$ INFN, Sezione di Pavia -- Via Bassi 6, 27100 Pavia, Italy.}\vspace{2mm}\\
{\small $^3$ Istituto Nazionale di Alta Matematica -- Sezione di Pavia, Via Ferrata, 5, 27100 Pavia, Italy.}\vspace{4mm}\\
 {\footnotesize  ~$^a$ mauro.carfora@pv.infn.it~,~$^b$ claudio.dappiaggi@unipv.it~,~$^c$ 
 nicolo.drago@unipv.it~,~$^d$ paolo.rinaldi01@universitadipavia.it }
 }

\date{May 21, 2019}

\begin{document}

\maketitle

\begin{abstract}
The perturbative approach to nonlinear Sigma models and the associated renormalization group flow are discussed within the framework of Euclidean algebraic quantum field theory and of the principle of general local covariance. In particular we show in an Euclidean setting how to define Wick ordered powers of the underlying quantum fields and we classify the freedom in such procedure by extending to this setting a recent construction of Khavkine, Melati and Moretti for vector valued free fields. As a by-product of such classification, we provide a mathematically rigorous proof that, at first order in perturbation theory, the renormalization group flow of the nonlinear Sigma model is the Ricci flow.
\end{abstract}


\section{Introduction}\label{Section: Introduction}

In the realm of geometric analysis, there are several open avenues of research which have benefited from results and models arising from classical and quantum field theory. One, if not the most prominent example is the Ricci flow, which has come to the fore in the past few years thanks to Perelman proof \cite{Perelman:2006un,Perelman:2006up} of the geometrization programme for three-dimensional manifolds due to Thurston \cite{Thurston}. Introduced in the mathematical literature in the early eighties by Hamilton \cite{Hamilton}, the Ricci flow has appeared independently in the context of quantum field theory mainly thanks to the early works of Friedan \cite{Friedan:1980jf,Friedan:1980jm} within the analysis of nonlinear Sigma models over two-dimensional Riemannian manifolds as source and with a Riemannian manifold of arbitrary dimension as target space. Despite the apparent distance between the two settings in which the Ricci flow first appeared, the mutual influences have been manifold and the field theoretical approach has been of inspiration for some of the ground breaking results of Perelman and for several analyses of the structural properties of such flow, see {\it e.g.} \cite{Carfora:2014pna}.

From the viewpoint of nonlinear Sigma models, the Ricci flow arises when, by considering a perturbative approach to the underlying Euclidean field theory, one studies at first order the renormalization group flow, see {\it e.g.} \cite{Carfora-17,Carfora:2017xia} and also \cite{Gaw99}. This is the main aspect on which we wish to focus in this paper and in particular we shall address the criticism towards such derivation of Ricci flow, which is often labelled not to be fully mathematically rigorous. 

In order to tackle this problem, we shall work within the framework of algebraic quantum field theory, a mathematically rigorous approach which was first formulated by Haag and Kastler \cite{Haag:1963dh}. Especially in the past few years it has been employed successfully to unveil and to characterize several structural properties of free and interacting quantum field theories, ranging from the formulation of the principle of general local covariance to a mathematically rigorous analysis of regularization and renormalization  -- see the recent reviews \cite{Brunetti:2015vmh,Rejzner:2016hdj}. Yet, the vast majority of the efforts have been addressed towards formulating and understanding the algebraic approach for field theories living on an underlying Lorentzian spacetime and thus many results and constructions are tied to such class of backgrounds. For example the quantization of a classical free field theory, the construction of an algebra of Wick polynomials or accounting interactions via a perturbative approach (including the ensuing renormalization procedure) are nowadays fully understood. Nonetheless a close scrutiny of all results unveils clearly that they rely on key structures which are tied to Lorentzian metrics. Notable instances of this statement are the realization of the canonical commutation relations in terms of advanced and retarded fundamental solutions associated to normally hyperbolic partial differential operators or the construction of Wick ordered quantum fields as a by-product of the existence of Hadamard sates, see \cite{Brunetti:2015vmh}. 

Yet there is no a priori obstruction to work within the algebraic framework while considering classical or quantum field theories which are living over a Riemannian manifold. Starting from the early seventies a few works in this direction have appeared in the literature \cite{Osterwalder:1973dx,Osterwalder:1974tc} and, despite most of the efforts went towards formulating algebraic quantum field theory on Lorentzian backgrounds, it is clear that most of the ideas, of the technique and of the structural aspects admit a well-defined Euclidean counterpart.
A notable example in this direction are the recent works by Keller on the formulation of Euclidean Epstein-Glaser renormalization \cite{Keller-09,Keller:2010xq}, see also \cite{Schlingemann:1998cw,Wald-79}.

Hence, motivated and inspired by these works, we decide to opt for a bottom-up approach towards the analysis of the nonlinear Sigma models which are at the heart of the Ricci flow. At a classical level such models are realized considering as kinematic configurations arbitrary smooth maps $\psi$ from a two-dimensional Riemannian manifold $(\Sigma,\gamma)$ into a target Riemannian background $(M,g)$ of arbitrary dimension. The dynamics is ruled by the stationary points of the so-called harmonic Lagrangian $\mathcal{L}_{\mathrm{H}}$ and considering its linearisation around an arbitrary configuration, we obtain a free field theory, which up to a source term, is governed by an elliptic operator $E$. This model is closely connected to string theory and its quantization from the algebraic viewpoint has been considered in \cite{Bahns-Rejzner-Zahn-14}.

As a starting point, we address the question of studying the quantization of the ensuing linearised theory. To this end, first we define the notion of an Euclidean locally covariant quantum field theory translating to a Riemannian framework the renown principle of general local covariance, formulated in a Lorentzian setting in \cite{Brunetti:2001dx}. This leads us naturally to identifying an Euclidean quantum field theory as a functor between a suitable category of background data into that of unital $*$-algebras which satisfies in addition a scaling hypothesis. Without entering into the technical details in the introduction, this requirement entails, that there exists an action of $\mathbb{R}_+:=(0,\infty)$ on the background data which, in turn, yields a corresponding isomorphism between the algebras of observables associated to each of the backgrounds constructed via such action.  In order for the model of our interest to fit in this scheme, we need as second step to show how to associate a $*$-algebra of locally covariant observables to the linear theory ruled by the operator $E$. 

To this end, we work with the functional formalism, which has been successfully applied to the Lorentzian setting, see for example \cite{Brunetti-Duetsch-Fredenhagen-09,Rejzner:2016hdj}. On the one hand we cannot follow slavishly these references, since we need to cope with several features which are tied to the Riemannian setting. On the other hand this approach has the net advantage that it allows to individuate in the ensuing algebra of observables a class of elements which could naturally be interpreted as Wick ordered powers of the underlying quantum field.

This observation leads to the second part of our paper in which we address the question of giving an abstract definition of Wick powers of an associated quantum field. This brings us to two relevant results. On the one hand, we characterize and classify the freedom which exists in constructing such polynomials, starting from the given definition. In tackling this problem, we extend to our framework the recent work of Khavkine, Melati and Moretti \cite{Khavkine-Melati-Moretti-17}, who have completely answered this question for vector valued Bosonic linear field theories, extending the seminal works of \cite{Hollands-Wald-01,Hollands-Wald-02,Hollands-Wald-05}. On the other hand we show that, since per assumption there exists an action of $\mathbb{R}_+$ on the background data, this induces a one-parameter family of Wick ordered powers of the underlying quantum field. It is worth mentioning that the formalism used in our work is strongly connected to the one of the recent monograph \cite{Herscovich}. Yet, in this reference, perturbative quantum field theory is presented from a very general viewpoint highlighting the minimal set of underlying assumptions which allow for the whole procedure to work. On the contrary we focus on a very specific scenario and in particular we follow a different approach to discuss the renormalization ambiguities of the underlying model. 

Subsequently, following the standard rationale used in perturbative algebraic quantum field theory \cite{Hollands-Wald-03}, it turns out that, to each coherent assignments of a one-parameter family of Wick polynomials, one associates a corresponding family of locally covariant Lagrangian densities, where parametric dependence is codified in the coupling constant, namely the metric of the target Riemannian manifold. As a last step, using the classification result of the ambiguities between two coherent assignment of Wick polynomials, we prove that such one-parameter family of metrics obeys the Ricci flow equation.

\vskip .3cm

The paper is organized as follows: In the next subsection we fix the notation and we introduce all the geometric and analytic building blocks necessary for our investigation, in particular the nonlinear sigma models, we are interested in and their linearisation. Section \ref{Sec: LCEFT} is devoted entirely to defining and to studying locally covariant Euclidean field theory. In particular in Subsection \ref{Sec: General Local Covariance}, first we introduce all the categories that we will be using and subsequently we give the formal definition of an Euclidean locally covariant theory, emphasizing in particular the so-called scaling hypothesis. In Subsection \ref{Subsection: subsection where we introduce the locally covariant theory of interest} we show instead that the linearisation on the nonlinear Sigma models, that we consider, fits in the framework formulated in Subsection \ref{Sec: General Local Covariance}. In Subsection \ref{Sec: LCO} we still focus on the model of our interest, defining the notion of locally covariant observables and studying their behaviour under the action of the scaling which is intrinsic in the definition of locally covariant Euclidean field theory. Finally in Subsection \ref{Section: Wick powers and ambiguities}, we generalize \cite{Khavkine-Melati-Moretti-17} to our setting defining first what is a family of Wick powers and then classifying the ambiguities existing in giving such definition. In Section \ref{Section: Renormalization and Ricci flow} we apply the results and the construction of Section \ref{Sec: LCEFT} to give a rigorous derivation of the Ricci flow from the perturbative renormalization group of the nonlinear Sigma models introduced in Subsection \ref{Section: Geometrical setting}.

\subsection{General Setting}\label{Section: Geometrical setting}
The goal of this section is to fix the notation and to introduce all the geometric and analytic building blocks necessary for our investigation. 

To start with, we consider two connected, oriented, Riemannian manifolds $(\Sigma,\gamma)$ and $(M,g)$ where $\dim M=D$, while $\dim\Sigma=D^\prime$.
Later we will consider only the case $D^\prime=2$.
In order to avoid confusion when dealing with the geometric structures associated to these backgrounds, we shall employ the convention that Greek ({\em resp.} Latin) indices are associated to quantities related to $\Sigma$ ({\em resp.} to $M$).
In addition, we denote with $\nabla^\Sigma,\nabla^M$ the Levi-Civita connections defined respectively on $T\Sigma,TM$.

\begin{remark}\label{Remark: on pull-back bundle and pull-back connection}
	For future convenience we recall the definition of pull-back bundle and of pull-back connection, {\it cf.} \cite{Husemoller}. Let $B\stackrel{\pi_B}{\longrightarrow}M$ be a vector bundle over $M$ -- typically $B=TM$ or $B=T^*M$ -- and let $\psi\in C^\infty(\Sigma;M)$.
	The pull-back bundle $\psi^*B$ is the vector bundle over $\Sigma$ defined by
	\begin{align}\label{Definition: pull-back bundle}
		\psi^*B:=\{(x,\xi)\in\Sigma\times B|\;\pi_B(\xi):=\psi(x)\}\qquad\pi_{\psi^*B}(x,\xi):=x\,,
	\end{align}
	From the definition it follows that $\widehat{\psi}\colon\psi^*B\ni(x,\xi)\mapsto\xi\in B|_{\psi(\Sigma)}$ is an injective morphism of vector bundles which lifts $\psi$, namely $\widehat{\psi}\circ\pi_{B}=\psi\circ\pi_{\psi^*B}$.
	This leads to an injective morphism of vector spaces $\psi^\sharp\colon\Gamma(B)\to\Gamma(\psi^*B)$ defined by $\psi^\sharp s(x):=(\widehat{\psi}^{-1}\circ s\circ\psi)(x)=(x,s\circ\psi(x))$ for all $s\in\Gamma(B)$.
	Considering the Levi-Civita connection $\nabla^B$ on $B$ the pull-back connection $\nabla^\psi$ on $\psi^*B$ is defined as follows \cite[App. C]{Milnor-Stasheff-76}.
	The push-forward $\mathrm{d}\psi\colon T\Sigma\to TM$ induces an injective homomorphism of sections $\mathrm{d}\psi^*\colon\Gamma(T^*M)\to\Gamma(T^*\Sigma)$.
	The pull-back connection $\nabla^\psi\colon\Gamma(B)\to\Gamma(T^*\Sigma\otimes B)$ is the unique one such that $\nabla^\psi\circ\psi^\sharp=(\mathrm{d}\psi^*\otimes\psi^\sharp)\circ\nabla^M$.
	Furthermore, per definition $\nabla^\psi\alpha:=\nabla^\Sigma\alpha$ for all $\alpha\in\Gamma(T^*\Sigma)$.
\end{remark}

On top of $\Sigma$, we consider kinematic configurations $\psi\in C^\infty(\Sigma;M)$ while dynamics is ruled by the stationary points of the so-called harmonic Lagrangian density $\mathcal{L}_{\mathrm{H}}$:
\begin{align}\label{Equation: harmonic Lagrangean density}
\mathcal{L}_{\mathrm{H}}[\psi,\gamma,g]:=\operatorname{tr}_\gamma(\psi^*g)\mu_\gamma
\stackrel{\textrm{loc.}}{=}g_{ab}\gamma^{\alpha\beta}(\mathrm{d}\psi)^a_\alpha(\mathrm{d}\psi)^b_\beta\mu_\gamma\,,
\end{align}
where $\mu_\gamma$ is the volume form induced by $\gamma$ while $\mathrm{d}\psi\colon T\Sigma\to TM$ is the push-forward along $\psi$.
In this paper we shall not work directly with \eqref{Equation: harmonic Lagrangean density}, rather we consider an expansion of $\mathcal{L}_{\mathrm{H}}$ up to the second order with respect to an arbitrary, but fixed background kinematic configuration $\psi$.
More precisely, for $\varphi\in\Gamma(\psi^*TM)$ let $\psi_\nu\colon\Sigma\to M$ be 
$\psi_\nu(x):=\exp_{\psi(x)}\big[\nu\varphi(x)\big]$, where $\exp_{\psi(x)}\colon T_{\psi(x)}M\to M$ is the exponential map at $\psi(x)\in M$, while $\nu\in I\subset\mathbb{R}$, where $I$ is an open subset of $\mathbb{R}$ which includes the origin -- \textit{cf.} remark \ref{Remark: on pull-back bundle and pull-back connection}.
The Taylor expansion of $\mathcal{L}_{\mathrm{H}}(\psi_\nu,\gamma,g)$ centred at $\nu=0$ yields 
\begin{align}
\label{Equation: second order expansion of harmonic Lagrangean}
\mathcal{L}(\psi_\nu,\gamma,g;\varphi)=
\mathcal{L}_{\mathrm{H}}(\psi,\gamma,g)+
\bigg[
\nu g(\varphi,Q(\psi))-
\frac{\nu^2}{2}\langle\varphi,E\varphi\rangle+
\frac{\nu^2}{2}h(\mathrm{Riem}(\varphi,\mathrm{d}\psi)\varphi,\mathrm{d}\psi)
\bigg]\mu_\gamma+O(\nu^3)\,,
\end{align}
where $\psi_0(x)\equiv\psi(x)$,  $h|_x\stackrel{\textrm{loc.}}{=}g_{ab}|_{\psi(x)}\gamma^{\alpha\beta}|_x\frac{\partial}{\partial x^\alpha}\otimes\frac{\partial}{\partial x^\beta}\otimes\mathrm{d}y^a\otimes\mathrm{d}y^b$, while the operator $E$ is defined by
\begin{align}\label{Equation: definition of the elliptic operator}
	E\colon\Gamma(\psi^*TM)\to\Gamma(\psi^*T^*M)\qquad
	E\varphi:=\operatorname{tr}_h(\nabla^\psi\circ\nabla^\psi\varphi)\,,
\end{align}
where $\nabla^\psi$ stands for the pull-back connection associated with $\nabla^M,\nabla^\Sigma$ on the pull-back bundle $\psi^*TM$ over $\Sigma$-- \textit{cf.} remark \ref{Remark: on pull-back bundle and pull-back connection}.
Finally the operator $Q\colon C^\infty(\Sigma,M)\to\Gamma(\psi^*T^*M)$ is a differential operator whose explicit form is inessential for what follows -- see \cite{Carfora-17} for details. For our purposes the following property is of paramount relevance:

\begin{lemma}\label{Lem:E is elliptic}
The operator $E$ is elliptic and its principal symbol coincides with that of $\widehat{E}\colon\Gamma(\psi^*TM)\to\Gamma(\psi^*T^*M)$, locally defined by
$(\widehat{E}\varphi)_a(x):=g_{ab}(\psi(x))\Delta_\gamma(\varphi^b(x))$ for $\varphi\in\Gamma(\psi^*TM)$, where $\Delta_\gamma$ is the Laplace-Beltrami operator built out of $\gamma$.
\end{lemma}

\begin{proof}
For any point $x\in\Sigma$ and for any local trivialization of $\psi^*TM$ centred at $x$, \eqref{Equation: definition of the elliptic operator} reads
	\begin{align*}
		(E\varphi)_b \stackrel{\textrm{loc.}}{=}
		g_{ab}\Delta_\gamma(\varphi^a)+
		g_{ab}\Delta_\gamma(\psi^\ell)\Gamma_{\ell c}^{\phantom{\ell c}a}[g]\varphi^c
		&+ g_{ab}\gamma^{\alpha\beta}\bigg[
		(\mathrm{d}\psi)^\ell_\alpha\frac{\partial}{\partial x^\beta}\big(\Gamma_{\ell c}^{\phantom{\ell c}a}[g]\varphi^c\big)
		+(\mathrm{d}\psi)^\ell_\alpha\Gamma_{\ell c}^{\phantom{\ell c}a}[g]\frac{\partial\varphi^c}{\partial x^\beta}\bigg]\\
		& + g_{ab}\gamma^{\alpha\beta}(\mathrm{d}\psi)^\ell_\beta(\mathrm{d}\psi)^p_\alpha\Gamma_{\ell c}^{\phantom{\ell c}a}[g]\Gamma_{pd}^{\phantom{pd}c}[g]\varphi^d\,.
	\end{align*}
	where $\varphi\in\Gamma(\psi^*TM)$.
	Considering the definition of principal symbol,
	\begin{align}
		\sigma_E(\mathrm{d}\zeta)\varphi:=\lim_{z\to+\infty}z^{-2}e^{-z\zeta}E(e^{z\zeta}\varphi)\,,\qquad
		\forall\zeta\in C^\infty(\Sigma),\varphi\in\Gamma(\psi^*TM)\,,
	\end{align}
	the sought statement follows.
\end{proof}

\begin{remark}\label{Remark: on the role of the Lagrangean}
	From now on the main object of our interest will be the expansion in \eqref{Equation: second order expansion of harmonic Lagrangean} and therefore $\psi,\gamma,g$ will be considered as parameters/background structures of the theory, whereas the role of kinematic configuration will be taken by $\varphi\in\Gamma(\psi^*TM)$.
	Observe that, in \eqref{Equation: second order expansion of harmonic Lagrangean}, $\langle\varphi,E\varphi\rangle\mu_\gamma$, plays the r\^ole of a kinetic term ruled by $E$, the elliptic operator \eqref{Equation: definition of the elliptic operator} associated to the Lagrangian $\mathcal{L}$.
\end{remark}

To conclude the section, we focus on the behaviour of the background structures under scaling and in particular we are interested in the \textit{engineer dimension} of ($\psi,\varphi,g$). The latter can be computed as follows: Consider the transformation
	\begin{align}\label{Equation: scaling transformation for gamma}
	\gamma\to\gamma_\lambda\,,\qquad(\gamma_\lambda)_{\alpha\beta}:=\lambda^{-2}\gamma_{\alpha\beta},\quad\lambda>0.
	\end{align}
	The engineer dimensions $\mathrm{d}_\psi,\mathrm{d}_\varphi,\mathrm{d}_g\in\mathbb{R}$ respectively of $\psi,\varphi,g$, appearing in \eqref{Equation: second order expansion of harmonic Lagrangean}, are the unique real numbers such that, if
	\begin{align}\label{Equation: scaling transformation for psi,varphi,g}
	\psi\to\psi_\lambda:=\lambda^{\mathrm{d}_\psi}\psi\,,\qquad
	\varphi\to\varphi_\lambda:=\lambda^{\mathrm{d}_\varphi}\varphi\,,\qquad
	g\to g_\lambda\,,\quad(g_\lambda)_{ab}:=\lambda^{\mathrm{d}_g}g_{ab}\,,
	\end{align}
	then the corresponding scaled Lagrangean density $\mathcal{L}(\psi_\lambda,\gamma_\lambda,g_\lambda;\varphi_\lambda)$ remains invariant, that is $$\mathcal{L}(\psi_\lambda,\gamma_\lambda,g_\lambda;\varphi_\lambda)=\mathcal{L}(\psi,\gamma,g;\varphi)\,.$$
	Considering \eqref{Equation: second order expansion of harmonic Lagrangean} and that $\mu_{\gamma_\lambda}=\lambda^{-D^\prime}\mu_\gamma$, a straightforward computation leads to
	\begin{align}\label{Equation: engineer dimension of psi,varphi,g}
	\mathrm{d}_\psi=\mathrm{d}_\varphi=0,\quad\mathrm{d}_g=D^\prime-2.
	\end{align}

\section{Locally Covariant Euclidean Field Theories}\label{Sec: LCEFT}

In this Section, our goal is twofold. On the one hand we want to introduce locally covariant Euclidean quantum field theories using the language of categories and of functionals as first introduced in \cite{Brunetti:2001dx} and \cite{Brunetti-Duetsch-Fredenhagen-09} respectively. Since, contrary to these seminal papers, we will also be interested in vector valued fields defined over Riemannian manifolds, we will also benefit greatly from \cite{Keller-09,Keller:2010xq} and from the recent works \cite{Khavkine:2014zsa, Khavkine-Melati-Moretti-17}. At the same time we want to reinterpret and to analyse the model introduced in Section \ref{Section: Geometrical setting} within this more general conceptual framework.

\subsection{General Local Covariance}\label{Sec: General Local Covariance}

Following \cite{Brunetti:2001dx}, the starting point of the principle of general local covariance consists of identifying a suitable set of categories which encode all necessary information of the underlying model. For the scopes of this paper the necessary ingredients are:

\begin{enumerate}
	\item $\mathsf{Bkg}_{\mathrm{D^\prime,D}}$, the category of background geometries, such that
	\begin{itemize}
		\item $\mathrm{Obj}(\mathsf{Bkg}_{\mathrm{D^\prime,D}})$ are pairs $(N;b)$ where $N\equiv (\Sigma, M)$ identifies a pair of smooth, connected, oriented manifolds, with $\dim\Sigma=D^\prime$ and $\dim M=D$, while $b\equiv(\psi,\gamma,g)$ codifies the background data, that is $\psi\in C^\infty(\Sigma;M)$ while $\gamma$ and $g$ are smooth Riemmannian metrics respectively on $\Sigma$ and $M$.
		\item $\operatorname{Ar}(\mathsf{Bkg}_{\mathrm{D^\prime,D}})$ are pairs $(\tau, t)$ where $\tau:\Sigma\to\widetilde{\Sigma}$ and $t:M\to\widetilde{M}$ are orientation preserving, isometric embeddings subject to the compatibility condition 
		\begin{align}\label{Equation: compatibility condition for arrows in BkgG}
		\widetilde{\psi}\circ\tau=t\circ\psi\, ,
		\end{align}
		where $\widetilde{\psi}\in C^\infty(\widetilde{\Sigma};\widetilde{M})$. If $\dim\Sigma=D^\prime=2$, with a slight abuse of notation we write $\mathsf{Bkg}\equiv\mathsf{Bkg}_{2,D}$ as $D$ plays no relevant role in our analysis
	\end{itemize} 
	\item $\mathsf{Alg}$ is the category whose objects are unital $*$-algebras, while the arrows are unit preserving, injective $*$-homomorphisms.
	\item $\mathsf{Vec}$ is the category whose objects are real vector spaces while the arrows are injective linear morphisms.
\end{enumerate}

\begin{remark}\label{Remark: definition of scaling in BkgG}
	Observe that, in comparison with \cite[Def. 3.4]{Khavkine-Melati-Moretti-17}, we adopt a slightly different definition of category of background geometries, adapted to the framework we consider. Nevertheless $\mathsf{Bkg}_{D^\prime,D}$ still enjoys the notable property of being {\em dimensionful} in the sense of \cite{Khavkine-Melati-Moretti-17}. In other words there exists an action of $\mathbb{R}_+:=(0,\infty)$ on $\operatorname{Obj}(\mathsf{Bkg}_{D^\prime,D})$
	\begin{align}\label{Equation: definition of scaling in BkgG}
	(N;b)=(\Sigma,M,\psi,\gamma,g)\to(N;b_\lambda)
	:=(\Sigma,M;\psi_\lambda,\gamma_\lambda,g_\lambda):=(\Sigma,M;\psi,\lambda^{-2}\gamma,\lambda^{D^\prime-2}g)\,,
	\end{align}
	which is preserved by the arrows of $\mathsf{Bkg}_{\mathrm{D^\prime,D}}$ and whose definition is tied to the engineer dimension of $\psi,g$ as per \eqref{Equation: engineer dimension of psi,varphi,g}.
\end{remark}

\begin{Definition}\label{Definition: locally covariant theory}
	An \textit{Euclidean locally covariant theory} is a covariant functor $\mathcal{A}\colon\mathsf{Bkg}_{\mathrm{D^\prime,D}}\to\mathsf{Alg}$ which satisfies the scaling hypothesis: For all $\lambda>0$, let $\mathcal{A}_\lambda:\mathsf{Bkg}_{D,D^\prime}\to\mathsf{Alg}$ be the covariant functor $\mathcal{A}\circ\rho_\lambda$, where $\rho_\lambda:\mathsf{Bkg}_{D,D^\prime}\to\mathsf{Bkg}_{D,D^\prime}$ is the functor defined as the identity on morphisms while on objects it acts as per \eqref{Equation: definition of scaling in BkgG}.
Then, for all $\lambda,\mu,\sigma>0$, there exists a natural isomorphism $\mathcal{A}_\mu\stackrel{\varsigma_{\lambda,\mu}}{\Longrightarrow}\mathcal{A}_\lambda$ (with inverse $\mathcal{A}_\lambda\stackrel{\varsigma^{-1}_{\lambda,\mu}}{\Longrightarrow}\mathcal{A}_\mu$) such that
	\begin{align}\label{Equation: properties of the scaling map}
	\varsigma_{\lambda,\mu}[N;b]=\varsigma_{\lambda,\sigma}[N;b]\circ\varsigma_{\sigma,\mu}[N;b_\lambda]\,,\qquad
	\varsigma_{\lambda,\lambda}[N;b]=\operatorname{Id}_{\mathcal{A}[N;b]}\,,
	\end{align}
	for all $(N;b)\in\operatorname{Obj}(\mathsf{Bkg}_{\mathrm{D^\prime,D}})$.
\end{Definition}

\begin{remark} For notational convenience we shall adopt the convention
$$\varsigma_{\lambda}[N;b]\doteq\varsigma_{1,\lambda}[N;b].\quad\forall [N;b]\in\operatorname{Obj}(\mathsf{Bkg})$$
\end{remark}

\begin{remark}
	The role of $\varsigma_\lambda$ is to ensure that the scaling $(N;b)\to(N;b_\lambda)$ is consistently implemented in the theory described by the functor $\mathcal{A}$. In turn $\mathcal{A}_\lambda$ can be interpreted as the functor describing the theory $\mathcal{A}$ \textit{at the  scale $\lambda$}, while the map $\varsigma_\lambda[N;b]\colon\mathcal{A}_\lambda[N;b]\to\mathcal{A}[N;b]$ codifies the rules needed to transform the same theory between different scales. This interpretation will have a significant r\^ole in our main result, see Theorem \ref{Theorem: renormalized Lagrangean density for Ricci flow application}.
\end{remark}

\subsection{Linearised nonlinear Sigma models as a locally covariant theory}\label{Subsection: subsection where we introduce the locally covariant theory of interest}

In this subsection we will show how to reformulate the model in Section \ref{Section: Geometrical setting} as an Euclidean locally covariant theory as per Definition \ref{Definition: locally covariant theory}. Therefore, henceforth $\dim\Sigma=2$ and we will only be interested in the category of background geometries $\mathsf{Bkg}\equiv\mathsf{Bkg}_{2,D}$. 

As starting point we focus on an arbitrary, but fixed, background geometry $(N;b)\in\operatorname{Obj}(\mathsf{Bkg})$ showing how to build the algebra $\mathcal{A}[N;b]$ associated with the Lagrangian \eqref{Equation: second order expansion of harmonic Lagrangean}, reformulating the whole construction in terms of categories only at a later stage. 

Let thus $(N;b)=(\Sigma,M;\psi,\gamma,g)\in\operatorname{Obj}(\mathsf{Bkg})$ be a background geometry and let $E\colon\Gamma(\psi^*TM)\to\Gamma(\psi^*T^*M)$ be the elliptic differential operator \eqref{Equation: definition of the elliptic operator}.
Since $E$ is elliptic as per Lemma \ref{Lem:E is elliptic}, it admits a parametrix $P\colon\Gamma_{\mathrm{c}}(\psi^*T^*M)\to\Gamma(\psi^*TM)$, {\it c.f.} \cite[Th. 4.4]{Wells}, unique up to smoothing operators such that
\begin{align}\label{Equation: defining property of parametrix}
	PE-\operatorname{Id}_{\Gamma_{\mathrm{c}}(\psi^*TM)}\in
	\Gamma(\psi^*TM\boxtimes\psi^*T^*M)\,,\qquad
	EP-\operatorname{Id}_{\Gamma_{\mathrm{c}}(\psi^*T^*M)}\in
	\Gamma(\psi^*TM\boxtimes\psi^*T^*M)\,.
\end{align}
\begin{remark}\label{Remark: on exterior tensor product and symmetrix sections}
	Throughout this paper we shall employ the following notation.
	Given a vector bundle $B\stackrel{\pi_B}{\longrightarrow}M$ and $k\in\mathbb{N}$ we denote with $\mathrm{S}^{\otimes k}B\stackrel{\pi_{\mathrm{S}^{\otimes k}B}}{\longrightarrow}M$ the $k$-th symmetric tensor product of $B$.
	With $B^{\boxtimes n}\stackrel{\pi_{B^{\boxtimes n}}}{\longrightarrow}M^n$ we identify the $k$-th exterior tensor product of $B$, that is, the vector bundle over $M^k$ with fibre $\pi_{B^{\boxtimes n}}^{-1}(x_1,\ldots,x_n)=\otimes_{\ell=1}^n\pi_B^{-1}(x_\ell)$.
	For a given $s\in\Gamma(B^{\boxtimes n})$ we denote with $[s]\in\Gamma(B^{\otimes n})$ the section obtained by considering the coinciding point limit of $s$, that is, the section obtained by pull-back of $s$ along the inclusion of the total diagonal $D_n\to M^n$ where $D_n:=\{(x_1,\ldots,x_n)\in M^n|\;x_1=\ldots =x_n\}$.
	The notation $\mathrm{S}\Gamma(B^{\boxtimes n})$ always refers to smooth sections over $B^{\boxtimes n}$ which are symmetrized with respect to the base points.
	Notice that if $s\in\mathrm{S}\Gamma(B^{\boxtimes n})$ then $[s]\in\Gamma(S^{\otimes n}B)$ .
	In particular, for $s\in\Gamma(B)$ we denote $s^{[\otimes] n}:=[s^{\otimes n}]\in\Gamma(\mathrm{S}^{\otimes n}B)$ where $s^{\otimes n}\in\mathrm{S}\Gamma(B^{\boxtimes n})$.
\end{remark}
Notice that the properties of the parametrices of being symmetric will play a distinguished r\^ole in the construction of a commutative algebra of observables -- \textit{cfr.} Definition \ref{Definition: locally covariant theory of interest}, in sharp contrast with the outcome of the same procedure in a Lorentzian setting where the dynamics is ruled by symmetric hyperbolic partial differential operators.

\begin{remark}\label{Remark: adjoint of parametrix}
	Since $E$ is formally self-adjoint, it follows that the formal adjoint of any parametrix $P$ is again a parametrix for $E$.
	We can therefore consider formally self-adjoint parametrices, whose space will be denoted with $\operatorname{Par}[N;b]$.
	Notice that, because of equation \eqref{Equation: defining property of parametrix}, $\operatorname{Par}[N;b]$ is an affine space modelled over $\mathrm{S}\Gamma(\psi^*TM^{\boxtimes 2})$.
\end{remark}

\begin{remark}\label{Remark: local Hadamard representation of parametrices}
	The parametrix $P$ admits locally a Hadamard representation which is constructed in detail in Appendix \ref{Appendix: Hadamard expansion for the parametrix of E} -- \textit{cf.} proposition \ref{Proposition: coincinding point limit of scaled and non-scaled Hadamard parametrix}.
	Here we recall the final result: Let $(x,x^\prime)$ be a pair of points lying in a suitably constructed convex, geodesic neighbourhood $\mathcal{O}\hookrightarrow\Sigma$ centred at $x$. Then the integral kernel of $P$ reads locally
	\begin{align}\label{Equation: local expression for the parametrix}
		P^{ab}(x,x^\prime)=H^{ab}(x,x^\prime)+W_P^{ab}(x,x^\prime)\,,\quad H^{ab}(x,x^\prime)=V^{ab}\log\frac{\sigma(x,x^\prime)}{\ell_H^2},
	\end{align}
	where $\sigma(x,x^\prime)$ is the halved squared geodesic distance between $x$ and $x^\prime$, $V\in\mathrm{S}\Gamma(\psi^*T\mathcal{O}^{\boxtimes 2})$ is a suitable symmetric tensor, while $H$ codifies the singular part of the parametrix and $W_P\in \mathrm{S}\Gamma(\psi^*T\mathcal{O}^{\boxtimes 2})$.
	It is important to keep in mind that, although \eqref{Equation: local expression for the parametrix} is meaningful only locally, one can use $[W_P]\in\Gamma(\mathrm{S}^{\otimes 2}\psi^*T\mathcal{O})$ together with a partition of unity argument in order to identify a globally defined $[W_P]\in\Gamma(\mathrm{S}^{\otimes 2}\psi^*TM)$ -- which does not depend on the chosen partition of unity -- where the subscript is used to highlight the dependence on the choice of the parametrix $P$ -- \textit{cf.} remark \ref{Remark: on pull-back bundle and pull-back connection}.
	The Hadamard representation and $[W_P]$ will be particularly important in the following construction as well as in the definition of locally covariant Wick powers -- \textit{cf.} Example \ref{Example: Wick powers exploited for Ricci flow}. 
\end{remark}	

In the following we will indicate with $P\in\operatorname{Par}[N;b]$ both the linear operator $P\colon\Gamma(\psi^*T^*M)\to\Gamma(\psi^*TM)$ and its associated distribution $P\in\mathrm{S}\Gamma_{\mathrm{c}}(\psi^*T^*M^{\boxtimes 2})'$ \cite[Thm.8.2.12]{Hormander-83}, the subscript $c$ indicating that we consider distributions over compactly supported test-sections. Recall that with $\Gamma_c$ we are implicitly assuming that the sections are symmetrized also with respect to the base points.	

Now we consider an arbitrary but fixed $P\in\operatorname{Par}[N;b]$, using it first to define a suitable unital $*$-algebra associated to the system whose dynamics is ruled by the operator $E$.
Secondly we show how to built an algebra which is independent from the chosen $P$.

\begin{Definition}\label{Definition: smooth local polynomial functionals}
	We denote with $\mathcal{P}_{\textrm{loc}}[N;b]$ the complex vector space of functionals $F\colon\Gamma(\psi^*TM)\to\mathbb{C}$ spanned by monomial functionals
	\begin{gather}
		F_{\omega_k}(\varphi):=\int_\Sigma\langle\varphi^{[\otimes] k},\omega_k\rangle\,,\quad
		\omega_k\in\Gamma_{\mathrm{c}}(\wedge^{\mathrm{top}}T^*\Sigma\otimes \mathrm{S}^{\otimes k}\psi^*T^*M)\,,\quad k\in\mathbb{N}\cup\{0\}\,,
	\end{gather}
	where $\varphi^{[\otimes]k}\in\Gamma(\mathrm{S}^{\otimes k}\psi^*TM)$ denotes the coinciding point limit of the symmetric tensor product $\mathrm{S}^{\otimes k}\varphi\in\mathrm{S}\Gamma(\psi^*TM^{\boxtimes k})$ -- \textit{cf.} remark \ref{Remark: on pull-back bundle and pull-back connection}.
	Moreover, $\Gamma_{\mathrm{c}}(\wedge^{\mathrm{top}}T^*\Sigma\otimes \mathrm{S}^{\otimes k}\psi^*T^*M)$ denotes the compactly supported sections of the vector bundle $\wedge^{\mathrm{top}}T^*\Sigma\otimes \mathrm{S}^{\otimes k}\psi^*T^*M$ -- here $\wedge^{\mathrm{top}}T^*\Sigma$ denotes the bundle of densities on $\Sigma$ -- while $\langle\varphi^k,\omega_k\rangle$ denotes the pairing $\varphi^{a_1}(x)\dots\varphi^{a_k}(x)\omega_{a_1\ldots a_k}(x)$.
	We refer to $\mathcal{P}_{\textrm{loc}}[N;b]$ as to the space of local polynomial functionals (with no derivatives of the configurations).
	For future convenience we set $\mathcal{P}[N;b]:=\bigoplus_{n\geq 0}\mathcal{P}_{\textrm{loc}}^{\otimes n}[N;b]$ with $\mathcal{P}_{\textrm{loc}}^{\otimes 0}\equiv\mathbb{C}$.
\end{Definition}
\begin{remark}\label{Remark: on properties of functionals}
	Notice that any $F\in\mathcal{P}_{\textrm{loc}}[N;b]$ enjoys the following remarkable properties which will be exploited in the forthcoming discussion:
	\begin{enumerate}
		\item $F$ is \textit{smooth}, namely, for all $\varphi,\varphi_1,...,\varphi_n\in\Gamma(\psi^*TM)$, $n\geq 1$,
		the $n$-th functional derivative $F^{(n)}[\varphi]$, defined as 
		\begin{align}
			\big\langle F^{(n)}[\varphi],\varphi_1\otimes\ldots\otimes\varphi_n\big\rangle:=
			\frac{\partial^n}{\partial s_1\dots\partial s_n}F\left(\varphi+\sum_{i=1}^{n}s_i\varphi_i\right)\bigg|_{s_1=\ldots s_n=0}\,,
		\end{align}
		identifies a symmetric, compactly supported distribution $F^{(n)}[\varphi]\in\mathrm{S}\Gamma(\psi^*TM^{\boxtimes n})'$.
		\item
		$F$ is \textit{compactly supported}, that is, $\overline{\bigcup_{\varphi}\textrm{supp}(F^{(1)}[\varphi])}$ is compact;
		\item 
		$F$ is \textit{local}, because for all $\varphi\in\Gamma(\psi^*TM)$, the $n$-th functional derivative $F^{(n)}[\varphi]$ is supported on the thin diagonal of $\Sigma^n=\underbrace{\Sigma\times...\times\Sigma}_n$, that is $\textrm{supp}(F^{(n)}[\varphi])\subset D_n:=\{(x_1,\ldots,x_n)\in\Sigma^n|\, x_1=\ldots=x_n\}$.
		Moreover $\textrm{WF}(F^{(n)}[\varphi])$ is transversal $T^*D_n=$, where $\textrm{WF}(F^{(n)}[\varphi])$ stands for the wave front set of $F^{(n)}[\varphi]$, \cite[Def. 8.1.2]{Hormander-83}.
	\end{enumerate}
For simplicity -- \textit{cf.} the proof of Proposition \ref{Proposition: Pw-product on fiber algebra} and Remark \ref{Remark: on the log-singularity of the parametrices} -- our definition excludes local polynomial functionals which contain derivatives of the configuration $\varphi$; the latter class will play no r\^ole in what follows.
\end{remark}

\begin{proposition}\label{Proposition: Pw-product on fiber algebra}
	The vector space $\mathcal{P}[N;b]$ consisting of smooth, local, polynomial functionals is an associative and commutative $*$-algebra if endowed with 
	the product
	\begin{align}\label{Equation: definition of Pw-product on fiber algebra}
		(F\cdot_PG)(\varphi)=\left(\mathcal{M}\circ\exp(\Upsilon_P)(F\otimes G)\right)(\varphi):=F(\varphi)G(\varphi)+\sum_{n\geq 1}\frac{1}{n!}\big\langle F^{(n)}[\varphi],P^{\otimes n}G^{(n)}[\varphi]\big\rangle\,,
	\end{align}
	where $P^{\otimes n}G^{(n)}[\varphi]\in\mathrm{S}\Gamma_{\mathrm{c}}(\psi^*T^*M^{\boxtimes n})'$ is the extension of $\underbrace{P\otimes...\otimes P}_n$ to $G^{(n)}[\varphi]$ according to \cite[Thm. 8.2.13]{Hormander-83}, while 
	$$\mathcal{M}(F\otimes G)(\varphi)=F(\varphi)G(\varphi),$$
	is the pointwise product. Here $\Upsilon_P$ is such that, for all $\varphi_1,\varphi_2\in\Gamma(\psi^*TM)$,
	\begin{align*}
		\Upsilon_P(F\otimes G)(\varphi_1,\varphi_2):=\langle F^{(1)}[\varphi_1],PG^{(1)}[\varphi_2]\rangle\,.
	\end{align*}
	The $*$-involution is completely characterized on $\mathcal{P}[N;b]$ by requiring $F^*(\varphi):=\overline{F(\varphi)}$.
	We denote with $\mathcal{F}_P[N;b]$ the $*$-algebra $(\mathcal{P}[N;b],\cdot_P,*)$.
\end{proposition}

\begin{proof}
  The first step consists of observing that \eqref{Equation: definition of Pw-product on fiber algebra} is well-defined.
  Convergence of the sum is guaranteed since, the functionals being polynomial, only a finite set of terms contributes.
  The only problem might arise from $\langle F^{(n)}[\varphi],P^{\otimes n}G^{(n)}[\varphi]\rangle$.
  Yet, in the case at hand, the singular behaviour of $F^{(n)}[\varphi],G^{(n)}[\varphi]$ is known -- \textit{cf.} remark \ref{Remark: on properties of functionals} -- while that of $P$ is that of the Hadamard parametrix $H$ whose local behaviour is logarithmic in the halved squared geodesic distance $\sigma$, see equations \eqref{Equation: ansaltz for Hadamard parametrix}.
  Hence, using that the scaling degree of $H$ is smaller than $2$ and that $F,G$ do not contain derivatives of $\varphi$, we can use \cite[Thm. 5.2]{Brunetti:1999jn} to infer that the contraction of $P^{\otimes n}\in\mathrm{S}\Gamma_{\mathrm{c}}(\psi^*T^*M^{\boxtimes 2n})'$ with $F^{(n)}\otimes G^{(n)}\in\mathrm{S}\Gamma(\psi^*TM^{\boxtimes 2n})'$ yields a well-defined distribution with compact support in $C^\infty(\Sigma)'$, which can thus be integrated against the constant function.
  To conclude the proof, we observe that associativity is guaranteed per construction while commutativity is a by-product of the fact that the parametrix $P$ is symmetric -- \textit{cf.} remark \ref{Remark: adjoint of parametrix}.
\end{proof}
	
\begin{remark}\label{Remark: action of parametrix on top densities}
	Equation \eqref{Equation: definition of Pw-product on fiber algebra} is well-defined as a consequence of \cite[Thm. 8.2.13]{Hormander-83} and of the extension of the parametrix $P$ as a distribution on $\Gamma_{\mathrm{c}}(\wedge^{\mathrm{top}}T^*\Sigma\otimes\psi^*T^*M)$.
	This is defined setting $P\alpha:=\mu_\gamma P(*^{-1}_\gamma\alpha)$ for $\alpha\in\Gamma_{\mathrm{c}}(\wedge^{\mathrm{top}}T^*\Sigma\otimes\psi^*T^*M)$ and 
	\begin{gather}
		\langle \alpha,P\beta\rangle:=\int_\Sigma\langle\alpha,\ast_\gamma P\alpha\rangle\,,
	\end{gather}
	for $\alpha,\beta\in\Gamma_{\mathrm{c}}(\wedge^{\mathrm{top}}T^*\Sigma\otimes\psi^*T^*M)$ while $\ast_\gamma$ denotes the Hodge operator $\ast_\gamma\colon\Gamma(\wedge^\bullet T^*\Sigma)\to\Gamma(\wedge^{\dim\Sigma-\bullet}T^*\Sigma)$.
\end{remark}

\begin{remark}\label{Remark: on the log-singularity of the parametrices}
	Notice that the previous Proposition strongly relies on the assumption $\dim\Sigma=2$ as well as on Definition \ref{Definition: smooth local polynomial functionals} of smooth polynomial local functionals without derivatives.
	As a matter of fact, for higher dimension or considering polynomial functionals including derivatives of the configuration $\varphi$, the contraction between $P^{\otimes n}$ and $F^{(n)}\otimes G^{(n)}$ would not be uniquely defined.
	In this case, different extensions exist as one can infer following \cite{Brunetti:1999jn} and thus one has to cope with \textit{families} of well-defined products $\cdot_P$.
	An application of these ideas has already been studied in \cite{FR12,FR13} in the context of gauge theories.
	The discussion of such scenario is behind the scopes of this paper and it is postponed to a future work \cite{Dappiaggi-Drago-Rinaldi-19}, see also \cite{Keller-09,Keller:2010xq} and \cite{Dang-14}.
\end{remark}

\begin{remark}\label{Rem: Wick_powers_fixed_P}
	It is worth observing that the algebra of local polynomial functionals $\mathcal{F}_P[N;b]$ already includes elements which can be interpreted as Wick powers of a field $\varphi$.
	As a concrete example, thought especially for a reader who is more familiar with the standard point splitting procedure, consider the functional $F_\omega(\varphi)=\int_\Sigma d\mu_\gamma\,\varphi^a(x)\varphi^b(x)\omega_{ab}(x)$ with $\omega\in\Gamma_{\mathrm{c}}(\mathrm{S}^{\otimes 2}\psi^*T^*M)$.
	One can in turn pick any sequence $(g_n)_a(x)(f_n)_b(x^\prime)$ with $f_n,g_n\in\Gamma_{\mathrm{c}}(\psi^*T^*M)$ for all $n\in\mathbb{N}$ such that, in the weak topology, $\lim\limits_{n\to\infty}(f_n)_a(x)(g_n)_b(x^\prime)=\omega_{ab}(x)\delta(x,x^\prime)$.
	As a consequence one can rewrite
	$$F_\omega(\varphi)=\lim_{n\to\infty} F^\prime_{f_n}(\varphi)F^\prime_{g_n}(\varphi)=\lim_{n\to\infty}\left(\left(F^\prime_{f_n}\cdot_P F^\prime_{g_n}\right)(\varphi)-P(f_n,g_n)\right),$$
	where $F^\prime_{f_n}(\varphi)=\int_\Sigma d\mu_\gamma\,f_n(x)\varphi(x)$.
The right hand side of this last chain of equalities translates in the functional language the standard expression yielding the definition of a Wick ordered, squared field via a point splitting procedure.
\end{remark}

Notice that $(N;b)\to\mathcal{F}_P[N;b]$ does \textit{not} identify an Euclidean locally covariant theory as per Definition \ref{Definition: locally covariant theory} due to the choice of an arbitrary $P\in\operatorname{Par}[N;b]$.
Our next goal is to overcome this hurdle and the first step in this direction consists of showing that, for a fixed object in $\mathsf{Bkg}$, all choices of $P$ are equivalent.
The following Proposition generalizes to the case in hand a well-known property, see \textit{e.g.} \cite[Lemma 2.1]{Hollands-Wald-01} for the counterpart in a Lorentzian setting.
Since the proof is identical, mutatis mutandis to that of \cite[Prop. 1.4.7]{Lindner:2013ila}, \cite[Prop. II.4]{Keller-09}, we omit it.

\begin{proposition}\label{Prop:alpha_map}
Let $(N;b)\in\operatorname{Obj}(\mathsf{Bkg})$ be arbitrary but fixed and let  $P,\widetilde{P}\in\operatorname{Par}[N;b]$.
Then the algebras $\mathcal{F}_P[N;b]$ and $\mathcal{F}_{\widetilde{P}}[N;b]$ are $*$-isomorphic, the $*$-isomorphism being realized by
\begin{align}\label{Equation: star-isomorphism between different fiber algebras}
	\alpha_P^{\widetilde{P}}\colon\mathcal{F}_{\widetilde{P}}[N;b]\to\mathcal{F}_P[N;b]\,,\quad
	(\alpha_P^{\widetilde{P}}F)(\varphi):=\bigg[\exp\big[\Upsilon_{P-\widetilde{P}}\big]F\bigg](\varphi)\,,
\end{align}
where 
\begin{align}\label{Equation: Gamma contraction operator}
\bigg[\exp\big[\Upsilon_{P-\widetilde{P}}\big]F\bigg](\varphi)&=\sum\limits_{n=0}^\infty\frac{1}{2^n n!}\langle(P-\widetilde{P})^{\otimes n},F^{(2n)}[\varphi]\rangle
\end{align}
and where $\Upsilon_{P-\widetilde{P}}$ is such that
\begin{equation*}
(\Upsilon_{P-\widetilde{P}}F)(\varphi):=\frac{1}{2}\big\langle P-\widetilde{P},F^{(2)}[\varphi]\big\rangle\,.
\end{equation*}
\end{proposition}

\noindent In view of this last Proposition we can recollect all $*$-algebras $\mathcal{F}_P[N;b]$ in a single object:
\begin{Definition}\label{Definition: bundle of fiber algebras}
	We call $\mathcal{E}[N;b]$ the bundle 
	\begin{align}\label{Equation: definition of the bundle of fiber algebras}
	\mathcal{E}[N;b]:=\bigcup_P\mathcal{F}_P[N;b],
	\end{align}
	with base space $\operatorname{Par}[N;b]$ and projection map $\pi_{\mathcal{E}[N;b]}(F_P):=P$.
\end{Definition}
\begin{remark}\label{Remark: lifted action on the bundle of fiber algebras}
	Notice that the action $P\to P+W$ of $\mathrm{S}\Gamma(\psi^*TM^{\boxtimes 2})$ on $\operatorname{Par}[N;b]$ can be lifted to $\mathcal{E}[N;b]$ via the $*$-isomorphism \eqref{Equation: star-isomorphism between different fiber algebras}:
	\begin{align}\label{Equation: lift to the bundle of the action on the base}
	\alpha_W(F_{P}):=\alpha_{P+W}^{P}F_{P}\,,\qquad\forall F_{P}\in\mathcal{E}[N;b]\,.
	\end{align}
\end{remark}

\begin{Definition}\label{Definition: locally covariant theory of interest}
	Let $\Gamma_{\textrm{eq}}(\mathcal{E}[N;b])$ be the the complex vector space of equivariant sections on $\mathcal{E}[N;b]$	
	\begin{align}\label{Equation: definition of locally covariant algebra of interest}
	\Gamma_{\textrm{eq}}(\mathcal{E}[N;b]):=\big\lbrace F\in\Gamma(\mathcal{E}[N;b])\;|\;F(P)=\alpha_{P}^{\widetilde{P}}F(\widetilde{P})\quad\forall P,\widetilde{P}\in\operatorname{Par}[N;b]\big\rbrace\,.
	\end{align}
	We denote with $\mathcal{A}[N;b]\equiv(\Gamma_{\textrm{eq}}(\mathcal{E}[N;b]),\cdot,*)$ the unital $*$-algebra with the pointwise product $\cdot_P$ as in \eqref{Equation: definition of Pw-product on fiber algebra} and with the fiberwise involution
	\begin{align}\label{Equation: product and involution on the locally covariant algebra of interest}
	(F\cdot G)(P):=F(P)\cdot_P G(P)\,,\qquad
	F^*(P):=F(P)^*\,,
	\end{align}
	for all $F,G\in\Gamma_{\textrm{eq}}(\mathcal{E}[N;b])$.
\end{Definition}

\noindent We can now prove that $\mathcal{A}[N;b]$ is the sought algebra.

\begin{remark}
	The algebra $\mathcal{A}[N;b]$ can also be read as a concrete realization of the unique (up to $*$-isomorphism) $*$-algebra $\mathfrak{A}[N;b]$ for which there exists a family of $*$-isomorphisms $\alpha_{P}\colon\mathfrak{A}[N;b]\to\mathcal{F}_P[N;b]$ for all $P\in\operatorname{Par}[N;b]$ such that, $\alpha_{P}=\alpha^P_Q\circ\alpha_Q$ for all $P,Q\in\operatorname{Par}[N;b]$. In addition we observe that the concrete algebras that we have constructed do not carry any topology. Following \cite{Brouder:2017ygf} one can bypass this limitation. Yet, working with topological $*$-algebras would not change significantly the properties and the constructions in this paper. On the contrary it would play a key role whenever one looks for algebraic states on $\mathcal{A}[N;b]$ and for the associated GNS representation. Since this issue goes well beyond the scope of this work, we shall not further comment about it.
\end{remark}

\begin{remark}\label{Rem: additional_functors}
In the following we need to specify a few additional functor. We call
\begin{align*}
&\Gamma\colon\operatorname{Obj}(\mathsf{Bkg})\to\operatorname{Obj}(\mathsf{Vec})\qquad
\Gamma[N;b]:=\Gamma(\psi^*TM)\,,\\
&\Gamma_{\mathrm{c}}\colon\operatorname{Obj}(\mathsf{Bkg})\to\operatorname{Obj}(\mathsf{Vec})\qquad
\Gamma_{\mathrm{c}}[N;b]:=\Gamma_{\mathrm{c}}(\psi^*T^*M)\,,\\
&\operatorname{Par}\colon\operatorname{Obj}(\mathsf{Bkg})\to\operatorname{Obj}(\mathsf{Vec})\qquad
[N;b]\mapsto\operatorname{Par}[N;b]\,.
\end{align*}
Let $(\tau,t)\in\operatorname{Ar}(\mathsf{Bkg})$ be an arrow from $(N;b)$ to $(\widetilde{N};\widetilde{b})$, that is, $\tau\colon\Sigma\to\widetilde{\Sigma}$ and $t\colon M\to\widetilde{M}$ are isometric, orientation preserving embeddings such that $\widetilde{\psi}\circ\tau=t\circ\psi$.
The map $\tau$ can be lifted to an isomorphism of vector bundles $\widehat{\tau}\colon\tau^*\widetilde{\psi}^*T\widetilde{M}\to\widetilde{\psi}^*T^*\widetilde{M}|_{\tau(\Sigma)}$ by setting $\widehat{\tau}(x,\xi):=(\tau(x),\xi)$ -- \textit{cf.} Remark \ref{Remark: on pull-back bundle and pull-back connection}.
In addition the compatibility condition \eqref{Equation: compatibility condition for arrows in BkgG} implies
\begin{align*}
\tau^*\widetilde{\psi}^*T\widetilde{M}=
(\widetilde{\psi}\circ\tau)^*T\widetilde{M}=
(t\circ\psi)^*T\widetilde{M}=\psi^*t^*T\widetilde{M}=\mathrm{d}t\circ\widehat{\psi}(\psi^*TM)\,,
\end{align*}
where $\mathrm{d}t\colon TM\to T\widetilde{M}$ is the push-forward along $t$, while $\widehat{\psi}\colon\psi^*TM\to TM|_{\psi(\Sigma)}$ has been defined in Remark \ref{Remark: on pull-back bundle and pull-back connection}.
The composition $\widehat{\tau}_t:=\widehat{\tau}\circ\mathrm{d}t\circ\widehat{\psi}\colon\psi^*TM\to\widetilde{\psi}^*T\widetilde{M}|_{\tau(\Sigma)}$ is thus an injective morphism of vector bundles and the same applies to $\widehat{\tau}_{t,\mathrm{c}}\colon\psi^*T^*M\to\widetilde{\psi}^*T^*\widetilde{M}|_{\tau(\Sigma)}$.
Hence, we can consider
\begin{align*}
&\Gamma[\tau,t]\colon\Gamma[\widetilde{N};\widetilde{b}]\to\Gamma[N;b]\qquad
\Gamma[\tau,t]\widetilde{\varphi}:=\widehat{\tau_t}^{-1}\circ\widetilde{\varphi}\circ\tau\,,\\
&\Gamma_{\mathrm{c}}[\tau,t]\colon\Gamma_{\mathrm{c}}[N;b]\to\Gamma_{\mathrm{c}}[\widetilde{N};\widetilde{b}]\qquad
\Gamma_{\mathrm{c}}[\tau,t]\omega:=\widehat{\tau}_{t,\mathrm{c}}\circ\omega\circ\tau^{-1}\,,\\
&\operatorname{Par}[\tau,t]\colon\operatorname{Par}[\widetilde{N};\widetilde{b}]\to\operatorname{Par}[N;b]\qquad
\operatorname{Par}[\tau,t]\widetilde{P}:=\Gamma[\tau,t]\circ \widetilde{P}\circ\Gamma_{\mathrm{c}}[\tau,t]\,.
\end{align*}
Notice that $\Gamma_{\mathrm{c}}[\tau,t]\omega$ is well-defined on account of the support properties of $\omega$, in particular $\Gamma_{\mathrm{c}}(\tau,t)\omega|_x=0$ if $x\notin\tau(\Sigma)$.
In other words $\Gamma,\operatorname{Par}$ (\textit{resp}. $\Gamma_{\mathrm{c}}$) are contravariant (\textit{resp}. covariant) functors from $\mathsf{Bkg}$ to $\mathsf{Vec}$.
\\
\end{remark}

\begin{proposition}\label{Proposition: A is a functor}
	For all $(N;b)\in\operatorname{Obj}(\mathsf{Bkg})$, let $\mathcal{A}\colon\mathsf{Bkg}\to\mathsf{Alg}$ be such that, for all $(N;b)\in\operatorname{Obj}(\mathsf{Bkg})$, $\mathcal{A}[N;b]$ is the unital $*$-algebra as per Definition \ref{Definition: locally covariant theory of interest}, while, for every $(\tau,t)\in\operatorname{Arr}(\mathsf{Bkg})$, $\mathcal{A}[\tau,t]\in\operatorname{Ar}(\mathsf{Alg})$ as
	\begin{align*}
	\mathcal{A}[\tau,t]\colon\mathcal{A}[N;b]\to\mathcal{A}[\widetilde{N};\widetilde{b}]\qquad
	\mathcal{A}[\tau,t]F:=\Gamma[\tau,t]^\sharp\circ F\circ\operatorname{\operatorname{Par}}[\tau,t]\,,
	\end{align*}
	where 	
	\begin{align*}
	\Gamma[\tau,t]^\sharp\colon\mathcal{F}_{\operatorname{Par}[\tau,t]\widetilde{P}}[N;b]\to\mathcal{F}_{\widetilde{P}}[\widetilde{N},\widetilde{b}]\qquad
	\Gamma[\tau,t]^\sharp F:=F\circ\Gamma[\tau,t]\,.
	\end{align*}
	Then $\mathcal{A}$ is a covariant functor.
\end{proposition}

\begin{proof}
	It suffices to observe that $\mathcal{A}$ is well-defined when acting on objects since $\mathcal{A}[N;b]=\Gamma_{\textrm{eq}}[N;b]$ is per construction a unital $*$-algebra, whereas the analysis in Remark \ref{Rem: additional_functors} entails that $\mathcal{A}[\tau,t]\in\operatorname{Arr}(\mathsf{Alg})$ for all $(\tau,t)\in\operatorname{Ar}(\mathsf{Bkg})$. Since all structure used to define $\mathcal{A}$ act covariantly, $\mathcal{A}$ is a covariant functor.
\end{proof}

\noindent In order to conclude that the functor $\mathcal{A}$ introduced in Definition \ref{Definition: locally covariant theory of interest} identifies an Euclidean locally covariant theory as per Definition \ref{Definition: locally covariant theory}, the scaling property remains to be discussed.
It is particularly important to stress the relation between such property and the local Hadamard representation of the parametrix -- \textit{cf.} remark \ref{Remark: different local representation of a parametrix under scaling}.
\begin{proposition}\label{Proposition: scaling map}
	Let $(N;b)=(\Sigma,M;\psi,\gamma,g)\in\operatorname{Obj}(\mathsf{Bkg})$, $(N;b_\lambda):=(\Sigma,M;\psi,\lambda^{-2}\gamma,g)\in\operatorname{Obj}(\mathsf{Bkg})$, for $\lambda>0$ -- \textmd{cf.} Remark \ref{Remark: definition of scaling in BkgG}.
	Let $\mathcal{A}[N;b], \mathcal{A}[N;b_\lambda]$ be the associated $*$-algebras as per definition \ref{Definition: locally covariant theory of interest}.
	Then the map
	\begin{align}\label{Equation: scaling map between parametrices}
		\mathrm{s}_\lambda\colon\operatorname{Par}[N;b]\to\operatorname{Par}[N;b_\lambda]\,,\qquad
		\mathrm{s}_\lambda P:=P_\lambda:=\lambda^{-2}P\,,
	\end{align}
	is an isomorphism of affine spaces. Furthermore, the map
	\begin{align}\label{Equation: scaling map}
		\varsigma_\lambda\colon\mathcal{A}[N;b_\lambda]\to\mathcal{A}[N;b]\,,\qquad(\varsigma_\lambda F)(P,\varphi):=\widehat{\mathrm{s}}_\lambda F(P_\lambda,\varphi)\,,
	\end{align}
	is an isomorphism of $*$-algebras such that condition \eqref{Equation: properties of the scaling map} holds true.
	Here $\widehat{\mathrm{s}}_\lambda\colon\mathcal{E}[N;b_\lambda]\to\mathcal{E}[N;b]$ denotes the unique lift of $\mathrm{s}_\lambda$ to an isomorphism of vector bundles such that $\pi_{\mathcal{E}[N;b]}\circ\widehat{\mathrm{s}}_\lambda=\mathrm{s}_\lambda\circ\pi_{\mathcal{E}[N;b_\lambda]}$.
\end{proposition}
\begin{proof}
	The first assertion is a direct consequence of the defining properties \eqref{Equation: defining property of parametrix} for $P\in\operatorname{Par}[N;b]$ and of the behaviour under the scaling $\gamma\to\lambda^{-2}\gamma$ of the operator $E$ defined in \eqref{Equation: definition of the elliptic operator}, that is $E\to\lambda^2 E$.
	The associated map $\widehat{\mathrm{s}}_\lambda\colon\mathcal{E}[N;b_\lambda]\to\mathcal{E}[N;b]$ is defined by $\widehat{\mathrm{s}}_\lambda(P,F):=(P_\lambda,F)$. Notice that this guarantees that $\varsigma_\lambda\colon\Gamma_{\textrm{eq}}[N;b_\lambda]\to\Gamma_{\textrm{eq}}[N;b]$ is well-defined and it satisfies condition \eqref{Equation: properties of the scaling map}.
	\\
	It remains to be shown that the map $\varsigma_\lambda$ defined in \eqref{Equation: scaling map} is a $*$-isomorphism between $\mathcal{A}[N;b_\lambda]$ and $\mathcal{A}[N;b]$.
	For that it is enough to show that $\varsigma_\lambda(F\cdot G)=\varsigma_\lambda F\cdot \varsigma_\lambda G$ for all $F,G\in\mathcal{A}[N;b_\lambda]$.
	Let $P\in\operatorname{Par}[N;b]$. A direct computation shows that
	\begin{align}\label{Equation: functional derivative of scaling map}
		(\varsigma_\lambda F)[P]^{(n)}=\widehat{\mathrm{s}}_\lambda F[P_\lambda]^{(n)}\,.
	\end{align}
	Moreover, we have that
	\begin{align}
		\big\langle(\varsigma_\lambda F)[P]^{(n)},P^{\otimes n}(\varsigma_\lambda G)[P]^{(n)}\big\rangle
		=\widehat{\mathrm{s}}_\lambda\big\langle F^{(n)}[P_\lambda],P^{\otimes n}G^{(n)}[P_\lambda]\big\rangle
		=\widehat{\mathrm{s}}_\lambda\big\langle F^{(n)}[P_\lambda],P_\lambda^{\otimes n}G^{(n)}[P_\lambda]\big\rangle\,,
	\end{align}
	where in the first equality we used equation \eqref{Equation: functional derivative of scaling map} and the fact that $\widehat{\mathrm{s}}_\lambda$ commutes with the contraction with $P$, while the second equality follows from the scaling properties of the Hodge operator $\ast_\gamma\to\lambda^2\ast_\gamma$ -- \textit{cf.} remark \ref{Remark: action of parametrix on top densities}.
	By inserting these results in the equations (\ref{Equation: definition of Pw-product on fiber algebra}-\ref{Equation: product and involution on the locally covariant algebra of interest}) for $\cdot$ the equality $\varsigma_\lambda(F\cdot G)[P]=[\varsigma_\lambda(F)\cdot\varsigma_\lambda(G)][P]$ follows.
\end{proof}
\begin{remark}\label{Remark: different local representation of a parametrix under scaling}
	With reference to Remark \ref{Remark: local Hadamard representation of parametrices}, we compare the local Hadamard expansion of the integral kernels $P^{bc}(x,x'),P_\lambda^{bc}(x,x')$ of the parametrices $P,P_\lambda=\lambda^{-2}P$.
	Notice that, although $P_\lambda=\lambda^{-2}P$, at the level of integral kernels it holds $P_\lambda^{bc}(x,x')=P^{bc}(x,x')$ because of the presence of the different volume forms $\mu_\gamma,\mu_{\gamma_{\lambda}}=\lambda^2\mu_\gamma$ -- \textit{cf.}, Remark \ref{Remark: action of parametrix on top densities}.
	Yet the Hadamard parametrix $H,H_\lambda$, appearing in equation \eqref{Equation: local expression for the parametrix}, and the associated smooth remainders $W_P,W_{P_\lambda}$ do change. More precisely, under the scaling $\gamma\to\lambda^{-2}\gamma$, in any geodesic neighbourhood $\mathcal{O}\subset\Sigma$, the singular part of the parametrix transforms as
	\begin{align*}
		H^{ab}(x,x^\prime)\to H^{ab}_\lambda(x,x^\prime)=V_\lambda^{ab}(x,x^\prime)\log\frac{\sigma_\lambda(x,x^\prime)}{\ell_H^2}\,,
	\end{align*}
	where $V_{\lambda}\in\mathrm{S}\Gamma(\psi^*TM^{\boxtimes 2})$. As a consequence, whenever we choose a parametrix $P$, which decomposes as $P^{ab}=H^{ab}+W^{ab}$, the counterpart associated with the rescaled Hadamard parametrix $H_\lambda$ reads
	\begin{equation}\label{Eq:rescaled_parametrix}
		P^{ab}(x,x^\prime)=H_\lambda^{ab}(x,x^\prime)+W_\lambda^{ab}(x,x^\prime)\,,
	\end{equation}
	As already highlighted in Remark \ref{Remark: local Hadamard representation of parametrices}, one can consider the coinciding point limit $x\to x^\prime$ to construct $[W_P]\in\Gamma(\mathrm{S}^{\otimes 2}\psi^*TM)$.
	On account of proposition \ref{Proposition: coincinding point limit of scaled and non-scaled Hadamard parametrix} of Appendix \ref{Appendix: Hadamard expansion for the parametrix of E}, under scaling the global section $[W_P]$ transforms as $[W_{P,\lambda}]^{bc}=[W_P]^{bc}-2g^{bc}\log\lambda$.
\end{remark}

Collecting definition \eqref{Definition: locally covariant theory of interest} and propositions \ref{Prop:alpha_map}-\ref{Proposition: A is a functor}-\ref{Proposition: scaling map} we have the following result which concludes the construction of an Euclidean locally covariant theory as per Definition \ref{Definition: locally covariant theory}.
\begin{proposition}\label{Prop: ELCFT}
	The functor $\mathcal{A}\colon\mathsf{Bkg}\to\mathsf{Alg}$ identifies an Euclidean locally covariant theory as per Definition \ref{Definition: locally covariant theory}.
\end{proposition}

\begin{remark}\label{Remark: on the general case of dimension D greater than two}
	Notice that the whole construction of the functor $\mathcal{A}$ profits from a simplification due to the dimensional restriction $D=2$.
	Indeed, as pointed out in remark \ref{Remark: on the log-singularity of the parametrices}, for $D>2$ the singularity behaviour of the parametrices $P$ would spoil the possibility to define the product $\cdot$ as per Definition \ref{Definition: locally covariant theory of interest} on the whole set of local polynomials $\mathcal{P}[N;b]$.
	In this latter case the product would have been defined on the subset $\mathcal{P}_{\textrm{reg}}[N;b]\subset\mathcal{P}[N;b]$ made of those elements $F\in\mathcal{P}[N;b]$ with smooth functional derivatives at any order.
	An extension procedure should be applied to define the product $\cdot$ among local polynomial functionals in the same spirit of \cite{Brunetti-Duetsch-Fredenhagen-09,Hollands-Wald-02,Keller-09}.
	We will refrain from describing such a procedure here, see however \cite{Dappiaggi-Drago-Rinaldi-19}.
\end{remark}

\noindent For the rest of this paper we will consider the local covariant theory $\mathcal{A}$ introduced in Definition \ref{Definition: locally covariant theory of interest}.

\subsection{Local covariance of observables and of quantum fields}\label{Sec: LCO}
In this section we will be especially interested in identifying a distinguished class of elements of $\mathcal{A}[N;b]$ yielding notion of locally covariant observable.
For future convenience we first introduce the following functors.
\begin{Definition}\label{Definition: action of the functor Gammak on Obj and Arr}
	Let $(N;b)=(\Sigma,M;\psi,\gamma,g),(\widetilde{N},\widetilde{b})=(\widetilde{\Sigma},\widetilde{M};\widetilde{\psi},\widetilde{\gamma},\widetilde{g})\in\operatorname{Obj}(\mathsf{Bkg})$ and let $(\tau,t)\in\operatorname{Ar}(\mathsf{Bkg})$ be an arrow from $(N;b)$ to $(\widetilde{N};\widetilde{b})$. We call  $C^\infty_{\mathrm{c}}\colon\mathsf{Bkg}\to\mathsf{Alg}$ the covariant functor
	\begin{align}\label{Equation: definition of smooth and compactly supported functions functor}
		C^\infty_{\mathrm{c}}[N;b]:=C^\infty_{\mathrm{c}}(\Sigma)\qquad C^\infty_{\mathrm{c}}[\tau,t]f:=f\circ\tau^{-1}
		\quad\forall f\in C^\infty_{\mathrm{c}}[N;b]\,.
	\end{align}
\end{Definition}	

Notice that $C^\infty_{\mathrm{c}}[\tau,t]f$ is well-defined on account of the support properties of $f$, in particular $C^\infty_{\mathrm{c}}[\tau,t]f(x)=0$ whenever $x\notin\tau(\Sigma)$. Let $\wedge^{\mathrm{top}}T^*\Sigma$ denotes the bundle of densities on $\Sigma$. Then we can 

\begin{Definition}\label{Def: new_functors}
	Let $k\in\mathbb{N}$.
	We call $\mathrm{S}\Gamma_{\mathrm{c}}^{k}\colon\mathsf{Bkg}\to\mathsf{Alg}$, the covariant functor such that, for all $(N;b)\in\operatorname{Obj}(\mathsf{Bkg})$
	\begin{align}\label{Equation: definition of the functor Gammak with compact support}
		\mathrm{S}\Gamma_{\mathrm{c}}^{k}[N;b]:=
		\bigoplus_{m=0}^\infty \mathrm{S}\Gamma_{\mathrm{c}}((\wedge^{\mathrm{top}}T^*\Sigma\otimes \mathrm{S}^{\otimes k}\psi^*T^*M)^{\boxtimes m})\,,\quad \mathrm{S}\Gamma_{\mathrm{c}}((\wedge^{\mathrm{top}}T^*\Sigma\otimes \mathrm{S}^{\otimes k}\psi^*T^*M)^{\boxtimes 0})\equiv\mathbb{C}
	\end{align}
	where $\mathrm{S}\Gamma_{\mathrm{c}}((\wedge^{\mathrm{top}}T^*\Sigma\otimes \mathrm{S}^{\otimes k}\psi^*T^*M)^{\boxtimes m})$ denotes the compactly supported symmetric sections of the $m$-th exterior tensor product of the vector bundle $\wedge^{\mathrm{top}}T^*\Sigma\otimes \mathrm{S}^{\otimes k}\psi^*T^*M$ -- \textit{cf.} Remark \ref{Remark: on exterior tensor product and symmetrix sections}. In addition, for all $(\tau,t)\in\operatorname{Arr}(\mathsf{Bkg})$, 
 $$\mathrm{S}\Gamma_{\mathrm{c}}^{k}[\tau,t]\colon\Gamma^k_{\mathrm{c}}[N;b]\to\Gamma^k_{\mathrm{c}}[\widetilde{N};\widetilde{b}],$$
 where $S\Gamma^k_{\mathrm{c}}[\tau,t]\omega:=\widehat{\tau}_{t,\mathrm{c}}^{m,k}\circ\omega\circ\tau^{-1}$ for all $\omega\in\mathrm{S}\Gamma_{\mathrm{c}}^k((\wedge^{\mathrm{top}}T^*\Sigma\otimes \mathrm{S}^{\otimes k}\psi^*T^*M)^{\boxtimes m})$. Here $\widehat{\tau}_{t,\mathrm{c}}\colon\psi^*T^*M\to\widetilde{\psi}^*T^*\widetilde{M}|_{\tau(\Sigma)}$ is an injective morphism of vector bundles-- \textit{cf.} proof of proposition \ref{Proposition: A is a functor}, which extends to a map $\widehat{\tau}_{t,\mathrm{c}}^{m,k}\colon (\wedge^{\mathrm{top}}T^*\Sigma\otimes \mathrm{S}^{\otimes k}\psi^*T^*M)^{\boxtimes m}\to (\wedge^{\mathrm{top}}T^*\widetilde{\Sigma}\otimes \mathrm{S}^{\otimes k}\widetilde{\psi}^*T^*\widetilde{M})^{\boxtimes m}|_{\tau(\Sigma)}$,  $m\in\mathbb{N}\cup\{0\}$, by considering a suitable symmetrized tensor product and pull-back for top-densities.
	 
\end{Definition}

\noindent Observe that similarly we define the contravariant functor $\Gamma^k\colon\mathsf{Bkg}\to\mathsf{Alg}$ as
	\begin{align}\label{Equation: definition of the functor Gammak}
		\mathrm{S}\Gamma^k[N;b]:=\bigoplus_{m=0}^\infty \mathrm{S}\Gamma_{\mathrm{c}}((\mathrm{S}^{\otimes k}\psi^*TM)^{\boxtimes m})\,,
	\end{align}
	The associated arrow $\Gamma^k[\tau,t]\colon\Gamma[\widetilde{N};\widetilde{b}]\to\Gamma^k[N;b]$ is obtained by considering the injective morphism of vector bundle $\widehat{\tau}_{t}\colon\psi^*TM\to\widetilde{\psi}^*T\widetilde{M}|_{\tau(\Sigma)}$ -- \textit{cf.} proof of Proposition \ref{Proposition: A is a functor} -- and its extension to $\widehat{\tau}_{t}^{m,k}\colon (\mathrm{S}^{\otimes k}\psi^*TM)^{\boxtimes m}\to (\mathrm{S}^{\otimes k}\widetilde{\psi}^*T\widetilde{M})^{\boxtimes m}|_{\tau(\Sigma)}$ for all $m\in\mathbb{N}\cup\{0\}$.
	The arrow $\mathrm{S}\Gamma^k[\tau,t]\colon\mathrm{S}\Gamma^k[\widetilde{N};\widetilde{b}]\to\mathrm{S}\Gamma^k[N;b]$ is then defined as $\mathrm{S}\Gamma^k[\tau,t]\widetilde{C}:=(\widehat{\tau}_{t}^{m,k})^{-1}\circ\widetilde{C}\circ\tau$ for all $\widetilde{C}\in\mathrm{S}\Gamma(( \mathrm{S}^{\otimes k}\widetilde{\psi}^*T\widetilde{M})^{\boxtimes m})$.

\begin{remark}\label{Rem: Gammak}
	Since, from time to time, it is convenient to focus on a fixed $m$-th symmetric, exterior tensor product of $\wedge^{\mathrm{top}}T^*\Sigma\otimes \mathrm{S}^{\otimes k}\psi^*T^*M$, we can work also with the covariant functors $\mathrm{S}\Gamma_{\mathrm{c}}^{k,m}\colon\mathsf{Bkg}\to\mathsf{Vec}$, $m\in\mathbb{N}\cup\{0\}$, such that for all $(N;b)\in\operatorname{Obj}(\mathsf{Bkg})$ and for all $(\tau,t)\in\operatorname{Ar}(\mathsf{Bkg})$,
	\begin{align}\label{Eq: Gammak}
		\mathrm{S}\Gamma_{\mathrm{c}}^{k,m}[N;b]:= \mathrm{S}\Gamma_{\mathrm{c}}((\wedge^{\mathrm{top}}T^*\Sigma\otimes \mathrm{S}^{\otimes k}\psi^*T^*M)^{\boxtimes m})\,,\qquad
		\mathrm{S}\Gamma_{\mathrm{c}}^{k,m}[\tau,t]:=\mathrm{S}\Gamma_{\mathrm{c}}^{k}[\tau,t]\big|_{\mathrm{S}\Gamma_{\mathrm{c}}^{k,m}[N;b]}\,.
	\end{align}
\end{remark}

In the following we will be mainly interested in functionals which are constructed out of compactly supported sections of an arbitrary $(\wedge^{\mathrm{top}}T^*\Sigma\otimes \mathrm{S}^{\otimes k}\psi^*T^*M)^{\boxtimes m}$.
We stress that our analysis is tied to the functor $\mathcal{A}$ introduced in Definition \ref{Definition: locally covariant theory of interest}, though the procedure can be extended to any Euclidean locally covariant theory as per Definition \ref{Definition: locally covariant theory}.
As a first step we need to relate the functors $\mathcal{A}$ and $\mathrm{S}\Gamma^{k}_{\mathrm{c}}$.
Inspired by \cite{Brunetti:2001dx}, we introduce
\begin{Definition}\label{Def:nat_transf}
	Let $k\in\mathbb{N}$ and let $\mathcal{A}\colon\mathsf{Bkg}\to\mathsf{Alg}$ and $\mathrm{S}\Gamma_{\mathrm{c}}^k\colon\mathsf{Bkg}\to\mathsf{Alg}$ be the functors as per Definition \ref{Definition: locally covariant theory of interest} and \ref{Def: new_functors} respectively.
	We call {\em locally covariant observable of degree $k$} a natural transformation $\mathcal{O}_k\colon\mathrm{S}\Gamma_{\mathrm{c}}^k\to\mathcal{A}$ \textit{i.e.}, for every $(N;b)\in\mathrm{Obj}(\mathsf{Bkg})$,  $\mathcal{O}_k[N;b]\colon\mathrm{S}\Gamma_{\mathrm{c}}^k[N;b]\to\mathcal{A}[N;b]$ is an arrow in $\mathsf{Alg}$ such that, for every $(\tau,t)\in\operatorname{Ar}(\mathsf{Bkg})$ mapping $(N;b)$ to $(\widetilde{N};\widetilde{b})$, it holds that 
	\begin{equation}\label{eq: comp_nat_trans}
	\mathcal{O}_k[\widetilde{N};\widetilde{b}]\circ\mathrm{S}\Gamma_{\mathrm{c}}^{k}[\tau,t]=\mathcal{A}[\tau,t]\circ\mathcal{O}_k[N;b]\,.
	\end{equation}
\end{Definition}

For concreteness, we underline that, for all $(N;b)\in\operatorname{Obj}(\mathsf{Bkg})$, $\mathcal{O}_k[N;b]$ can be read as an algebra-valued distribution, that is, for all $P\in\operatorname{Par}[N;b]$, for all $\varphi\in\Gamma(\psi^*TM)$ and for all $m\in\mathbb{N}$,
\begin{align}\label{eq:notation}
	\mathcal{O}_k[N;b](\bullet,P,\varphi)\colon\mathrm{S}\Gamma_{\mathrm{c}}((\wedge^{\mathrm{top}}T^*\Sigma\otimes \mathrm{S}^{\otimes k}\psi^*T^*M)^{\boxtimes m})\ni\omega_m\mapsto
	\mathcal{O}_k[N;b](\omega_m,P,\varphi)\in\mathbb{C}\,,
\end{align}
defines a distribution $\mathcal{O}_k[N;b](\bullet,P,\varphi)\in\mathrm{S}\Gamma_{\mathrm{c}}((\wedge^{\mathrm{top}}T^*\Sigma\otimes \mathrm{S}^{\otimes k}\psi^*T^*M)^{\boxtimes m})'$.
Henceforth we will follow \eqref{eq:notation} writing for notational simplicity $\mathcal{O}_k[N;b](\omega,P,\varphi)$ in place $\big[\mathcal{O}_k[N;b](\omega)\big](P,\varphi)$.

\begin{remark}
	The previous definition can be generalized by substituting the functor $\mathrm{S}\Gamma_{\mathrm{c}}^{k}$ with an arbitrary functor $F\colon\mathsf{Bkg}\to\mathsf{Alg}$.
	In this case we still call local and covariant observable any natural transformation $\mathcal{O}\colon F\to\mathcal{A}$ -- \textit{cf.} equation \eqref{Equation: locally covariant interacting Lagrangean density for a fixed family of Wick powers} in section \ref{Section: Renormalization and Ricci flow}.
\end{remark}

\begin{remark}\label{Remark: on the alg-morphism property of locally covariant observables}
	Notice that, since $\mathcal{O}_k[N;b]\in\operatorname{Ar}(\mathsf{Alg})$, for all $m\in\mathbb{N}$ and $\omega_1\otimes\ldots\otimes\omega_m\in\mathrm{S}\Gamma_{\mathrm{c}}^{k,m}[N;b]$ with $\omega_j\in\mathrm{S}\Gamma_{\mathrm{c}}^{k,1}[N;b]$ for all $j=1,\ldots,m$ it holds
	\begin{align*}
		\mathcal{O}_k[N;b](\omega_1\otimes\dots\otimes\omega_m)=\mathcal{O}_k[N;b](\omega_1)\cdots\mathcal{O}_k[N;b](\omega_m)\,.
	\end{align*}
	This property implies that a locally covariant observable as per Definition \ref{Def:nat_transf} is known once it is known its value on degree $m\in\{0,1\}$.
	In this sense, a locally covariant observable consists of a locally covariant polynomial in the field -- $\mathcal{O}_k[N;b]$ at degree $m=1$ -- together with its powers according to the product $\cdot$ of $\mathcal{A}[N;b]$ -- see Definition \ref{Definition: locally covariant theory of interest}.
	As already stressed in Remark \ref{Remark: on the general case of dimension D greater than two} these observations depend crucially on the dimensional restriction $D=2$.
	For generic $D$, the identification of a locally covariant observable $\mathcal{O}_k$ can be interpreted as: (a) the identification of a local and covariant polynomial functional in the field configuration $\varphi$, namely of $\mathcal{O}_k[N;b]$ at degree $m=1$; (b) the identification of an extension of the product $\cdot$, which allows to define the product between $\mathcal{O}_k[N;b]$ with itself.
\end{remark}

\begin{Example}\label{Ex: loc. cov. obs.}
	Consider $(N;b)\in\operatorname{Obj}(\mathsf{Bkg})$ and let $\Phi[N;b]\colon\mathrm{S}\Gamma_{\mathrm{c}}^1[N;b]\to\mathcal{A}[N;b]$ be defined as follows.
	If $\omega_1\in\mathrm{S}\Gamma_{\mathrm{c}}^{1,1}[N;b]$ then $\Phi[N;b](\omega_1)$ is the linear functional such that, for all $(P,\varphi)\in\operatorname{Par}[N;b]\times\Gamma(\psi^*TM)$,
	\begin{align}\label{Equation: definition of locally covariant Phi}
		\Phi[N;b](\omega_1,P,\varphi):=\int_\Sigma\langle\omega_1,\varphi\rangle\,.
	\end{align}
	As pointed out in Remark \ref{Remark: on the alg-morphism property of locally covariant observables}, this fixes completely $\Phi[N;b]$ on the whole $\mathrm{S}\Gamma_{\mathrm{c}}^{1,1}[N;b]$. 
	Let now $(\tau,t)\in\operatorname{Ar}(\mathsf{Bkg})$ be a mapping from $(N;b)$ to $(\widetilde{N},\widetilde{b})\in\mathrm{Obj}(\mathsf{Bkg})$.
	To conclude that $\Phi$ is a natural transformation, we need to show that $[\mathcal{A}(\tau,t)]\circ\Phi[N;b]=\Phi[\widetilde{N},\widetilde{b}]\circ\mathrm{S}\Gamma_{\mathrm{c}}^1(\tau,t)$.
	This is a direct consequence of the definition, as one can readily infer, since, for every $\omega_1\in\mathrm{S}\Gamma_{\mathrm{c}}^{1,1}[N;b]$ and $\widetilde{P}\in\operatorname{Par}(\widetilde{N};\widetilde{b})$, $\widetilde{\varphi}\in\Gamma(\widetilde{\psi}^*T\widetilde{M})$
	\begin{align*}
		\bigg[\mathcal{A}[\tau,t]\Phi[N;b](\omega_1)\bigg](\widetilde{P},\widetilde{\varphi})&=
		\Phi[N;b](\omega_1,\operatorname{Par}[\tau,t]\widetilde{P},\mathrm{S}\Gamma^{1}[\tau,t]\widetilde{\varphi})=
		\int_{\widetilde{\Sigma}}\langle\widetilde{\varphi},\mathrm{S}\Gamma_{\mathrm{c}}^{1}[\tau,t]\omega_1\rangle\\&=
		\Phi[\widetilde{N},\widetilde{b}](\mathrm{S}\Gamma_{\mathrm{c}}^1(\tau,t)\omega_1,\widetilde{P},\widetilde{\varphi})\,,
	\end{align*}
	where $\operatorname{Par}[\tau,t]\widetilde{P}\in\operatorname{Par}[\widetilde{N};\widetilde{b}]$ has been defined in the proof of proposition \ref{Proposition: A is a functor}.
In view of its definition and of its properties $\Phi[N;b]$ identifies a locally covariant observable of degree $1$ to which we refer as a {\em locally covariant quantum field}.	
\end{Example}

\begin{Example}\label{Ex: Phi^2}
	In order to define powers of a locally covariant quantum field, which could be interpreted also as locally covariant observables, the starting point is Remark \ref{Rem: Wick_powers_fixed_P}.
	Here a candidate for a well-defined Wick ordered, squared field is introduced, but the definition depends on the choice of a parametrix $P$, a procedure which is intrinsically non locally covariant.
	In order to bypass this hurdle, constructing at the same time an equivariant section of $\mathcal{E}[N;b]$, we need to rely on the Hadamard representation of any parametrix $P$ as in Equation \eqref{Equation: local expression for the parametrix}.
	As outlined in Example \ref{Remark: local Hadamard representation of parametrices}, we can use such representation to identify from each parametrix $P$, $[W_P]\in\Gamma(\mathrm{S}^{\otimes 2}\psi^*TM)$.
	Bearing in mind this information, consider $(N;b)\in\operatorname{Obj}(\mathsf{Bkg})$ and let $\Phi^2[N;b]\colon\mathrm{S}\Gamma^{2}_{\mathrm{c}}[N;b]\to\mathcal{A}[N;b]$ be defined as follows.
	If $\omega_1\in\mathrm{S}\Gamma_{\mathrm{c}}^{2,1}[N;b]$ then $\Phi^2[N;b](\omega_1)$ is the linear functional such that, for all $(P,\varphi)\in\operatorname{Par}[N;b]\times\Gamma(\psi^*TM)$,
	\begin{align}\label{eq: Equiv. Phi^2}
		\Phi^2[N;b](\omega_1,P,\varphi):=\int_\Sigma\langle\varphi^{[\otimes]2}+[W_P],\omega_1\rangle\,,
	\end{align}
	where $\varphi^{[\otimes]2}\in\Gamma(\mathrm{S}^{\otimes 2}\psi^*TM)$ -- \textit{cf.} remark \ref{Remark: on exterior tensor product and symmetrix sections}.
	In order to realize that $\Phi^2$ identifies a locally covariant observable of degree $1$, it suffices to proceed as in Example \ref{Ex: loc. cov. obs.} and thus we shall not dwell into the details. It is important to observe that \eqref{eq: Equiv. Phi^2} is a possible realization of a Wick power of $\Phi$, but it is not the unique one. We will discuss this issue in detail in the next section.
\end{Example}

For later convenience we introduce a notion which intertwines locally covariant observables with scaling yielding as a by-product an abstract  notion of engineer dimension which matches the one discussed at the end of Section \ref{Section: Geometrical setting}.

\begin{Definition}\label{Definition: scaled locally covariant observable}
	Let $k\in\mathbb{N}$, and let $\mathcal{O}_k\colon\mathrm{S}\Gamma^{k}\to\mathcal{A}$ be a locally covariant observable of degree $k$ as per Definition \ref{Def:nat_transf}.
	For any $[N;b]\in\mathrm{Obj}(\mathsf{Bkg})$ we call {\em rescaled locally covariant observable} at scale $\lambda>0$, $S_\lambda\mathcal{O}_k$ the locally covariant observable defined by
	\begin{align}\label{Equation: scaled locally covariant observable}
		\big(S_\lambda\mathcal{O}_k\big)[N;b]:=\mathcal{O}_k[N;b_\lambda]\,,
	\end{align}
	where $[N;b_\lambda]$ is defined in \eqref{Equation: definition of scaling in BkgG}.
	In addition we say that $\mathcal{O}_k[N;b]$ has engineering dimension $\mathrm{d}_{\mathcal{O}_k}\in\mathbb{R}$ if
	\begin{align}\label{Equation: defining equation for engineer dimension}
		(S_\lambda\mathcal{O}_k)[N;b](\omega_m)=\lambda^{\mathrm{d}_{\mathcal{O}_k}m}\mathcal{O}_k[N;b](\omega_m)\,,	
	\end{align}
	holds for all $[N;b]\in\operatorname{Obj}(\mathsf{Bkg})$ and $\omega_m\in\mathrm{S}\Gamma_{\mathrm{c}}^{k,m}[N;b]$.
	On the contrary we say that $\mathcal{O}_k$ scales {\em almost homogeneously} with dimension $\kappa\in\mathbb{R}$ and order $\ell\in\mathbb{N}$ if
	\begin{align}\label{Eq: scaling alm_hom}
		S_\lambda\mathcal{O}_k[N;b](\omega_m)=\lambda^{\kappa m}\mathcal{O}_k[N;b](\omega_m)+
		\lambda^{\kappa m}\sum_{j=0}^\ell\log(\lambda)^j\mathcal{O}_j[N;b](\omega_m)\,,
	\end{align}
	holds for all $[N;b]\in\operatorname{Obj}(\mathsf{Bkg})$ and $\omega_m\in\mathrm{S}\Gamma_{\mathrm{c}}^{k,m}[N;b]$ where, for all $j\in\{0,\ldots,\ell\}$, $\mathcal{O}_j$ is a locally covariant observables which scales almost homogeneously with degree $\kappa$ and order $\ell-j$.
\end{Definition}
\begin{remark}\label{Remark: scaled locally covariant observable in D-dim}
	We stress, that while our analysis could be slavishly applied to models for which the dimension of $\Sigma$ is arbitrary, Definition \ref{Definition: scaled locally covariant observable} relies on $\dim\Sigma=2$. In the general case, \eqref{Equation: scaled locally covariant observable} should be modified as follows
	\begin{align*}
    \big(S_\lambda\mathcal{O}\big)[N;b]:=\varsigma_\lambda[N;b]\circ\mathcal{O}[N;b_\lambda]\,,
    \end{align*}
    where $\dim\Sigma=D$ while $\varsigma_{\lambda}[N;b]\colon\mathcal{A}[N;b_\lambda]\to\mathcal{A}[N;b]$ is the scaling transformation introduced in Definition \ref{Definition: locally covariant theory}.
\end{remark}

\noindent The engineering dimension can be computed explicitly in many notable instances:

\begin{Example}\label{Exam: scaling}
	Consider the locally covariant observables defined in Example \ref{Ex: loc. cov. obs.} via the natural transformations $\Phi$.
	Putting together \eqref{Equation: definition of locally covariant Phi} and \eqref{Equation: definition of scaling in BkgG}, one can compute that $S_\lambda\Phi=\Phi$ that is the engineering dimension of $\Phi$ is $0$.
	
	At the same time, if we consider the locally covariant observable $\Phi^2$ as per Example \ref{Ex: Phi^2}, in order to evaluate its behaviour under scaling we need to take into account Remark \ref{Remark: different local representation of a parametrix under scaling} according to which $[W_{P,\lambda}]^{bc}=[W_P]^{bc}-2g^{bc}\log\lambda$. Hence, for all $\lambda>0$
	$$(S_\lambda\Phi^2)=\Phi^2+\mathcal{V}\log\lambda,$$
	where $\mathcal{V}$ is the locally covariant observable of degree $0$ such that 
	$$\mathcal{V}[N;b](\omega_1,P,\varphi)=-2\int\limits_\Sigma \langle g^\sharp,\omega_1\rangle\,,\qquad\omega_1\in\mathrm{S}\Gamma_{\mathrm{c}}^{2,1}[N;b]\,.$$
	In other words $\Phi^2$ scales almost homogeneously with dimension $0$ and order $1$.
\end{Example}	

\subsection{Wick ordered powers of quantum fields}\label{Section: Wick powers and ambiguities}
Following our previous analysis, in this section we address the issue of Wick ordering in order to construct, for any $[N;b]\in\mathrm{Obj}(\mathsf{Bkg})$, well-defined algebra valued distributions, which can be read as locally covariant powers of the underlying, locally covariant quantum field $\Phi$ as the one introduced in Example \ref{Ex: loc. cov. obs.}. 

Although, in the Lorentzian framework, this is an overkilled topic starting from the seminal work \cite{Hollands-Wald-01}, here we will be mainly interested in the Euclidean setting and in vector-valued fields. For this reason we shall follow mainly the rationale used in \cite{Khavkine-Melati-Moretti-17}. In particular, tackling the problem of Wick ordering can be divided in two separate issues, the first concerning the existence of a well-defined ordering scheme, the second addressing the question of classifying the possible ambiguities in the construction of Wick ordered observables, while keeping track of local covariance. 

In the following we give an abstract definition of Wick ordered powers of a quantum field adapting to the case in hand \cite[Def. 5.2]{Khavkine-Melati-Moretti-17}.

\begin{Definition}\label{Definition: Wick power associated with a locally covariant observable}
	Let $\Phi$ be a locally covariant observable defined in Example \ref{Ex: loc. cov. obs.}.
	A family of Wick powers associated to $\Phi$ is a family of natural transformations $\Phi^\bullet=\{\Phi^k\}_{k\in\mathbb{N}}$ with $\Phi^k\colon\mathrm{S}\Gamma_{\mathrm{c}}^k\to\mathcal{A}$ such that it holds
	\begin{enumerate}
		\item
		For all $k\in\mathbb{N}\cup\{0\}$, $\Phi^k$ is a locally covariant observable which scales almost homogeneously with dimension $\kappa=0$.
		\item 
		If $k=1$, $\Phi^1=\Phi$ while, if $k=0$ $\Phi^0:=1_{\mathcal{A}}$, where for all $(N;b)\in\mathrm{Obj}(\mathsf{Bkg})$ and for all $z\in\mathbb{C}$, $1_{\mathcal{A}}[N;b](z):=z\,1_{\mathcal{A}[N;b]}$, the right hand side of this equality being the identity element of $\mathcal{A}[N;b]$.
		\item For all $k\in\mathbb{N}\cup\{0\}$, it holds, that, for all $(N;b)\in\operatorname{Obj}(\mathsf{Bkg})$, $\omega_1\in\mathrm{S}\Gamma_{\mathrm{c}}^{k,1}[N;b]$, $P\in\operatorname{Par}[N;b]$ and $\varphi_1,\varphi_2\in\Gamma(\psi^*TM)$,
		then, 
		\begin{align}\label{Equation: inductive condition on derivative for Wick powers}
			\big\langle\Phi^k[N;b](\omega_1,P)^{(1)}[\varphi_1],\varphi_2\big\rangle=
			k\,\Phi^{k-1}[N;b](\varphi_2\lrcorner\omega_1,P,\varphi_1)\,,
		\end{align}
		where $\varphi_2\lrcorner\omega_1\in\mathrm{S}\Gamma_{\mathrm{c}}^{k-1,1}[N;b]$ denotes the section which reads locally $\varphi_2^{a_1}(\omega_1)_{a_1\ldots a_k}$, while the superscript $(1)$ refers to the functional derivative as per Definition \ref{Definition: smooth local polynomial functionals}.
		\item 
		Let $d\in\mathbb{N}$ and let $(N;b_s)\in\operatorname{Obj}(\mathsf{Bkg})$ be such that $\{b_s=(\psi,\gamma_s,g_s)\}_{s\in\mathbb{R}^d}$ is a smooth, compactly supported $d$-dimensional family of variations of $b=(\psi,\gamma,g)$ as per Definition \ref{Definition: smooth compactly supported d-dimensional family of variations}.
		For all smooth family $\{P_s\}_{s\in\mathbb{R}^d}$ where $P_s\in\operatorname{Par}(N,b_s)$ for all $s\in\mathbb{R}^d$, let $\mathcal{U}_k\in\Gamma_{\mathrm{c}}(\pi_d^*\mathrm{S}^{\otimes k}\psi^*T^*M)'$ be the distribution on the pullback bundle $\pi_d^*\mathrm{S}^{\otimes k}\psi^*T^*M$ over the base space $\mathbb{R}^d\times\Sigma$ -- here $\pi_d\colon\mathbb{R}^d\times\Sigma\to\Sigma$ denotes the canonical projection -- defined by
		\begin{align}\label{Equation: distribution associated with Wick powers}
		\mathcal{U}_k(\chi\otimes\omega_1):=\int_{\mathbb{R}^d}\mathrm{d}s\,\Phi^k[N;b_s](\omega_1,P_s,0)\chi(s)\,,\qquad
		\omega_1\in\mathrm{S}\Gamma_{\mathrm{c}}^{k,1}[N;b]\,,\chi\in C^\infty_{\mathrm{c}}(\mathbb{R}^d)\,.
		\end{align}
		We require that, for all $k\in\mathbb{N}\cup\{0\}$.
		\begin{align}\label{Equation: regularity assumption on Wick powers}
		\textrm{WF}(\mathcal{U}_k)=\emptyset\,,
		\end{align}
		where $\textrm{WF}(\mathcal{U}_k)$ denotes the wavefront set of $\mathcal{U}_k$, \cite[Def. 8.1.2]{Hormander-83}.
	\end{enumerate}
\end{Definition}
\begin{remark}\label{Remark: on the regularity property}
	Notice that condition 4 in Definition \ref{Definition: Wick power associated with a locally covariant observable} exploits a smooth, compactly supported $d$-dimensional family of variations $(\gamma_s,g_s)$ of the metrics $\gamma$ and $g$ while the background configuration $\psi$ has been fixed.
	The choice of a smooth family of parametrices $\{P_s\}_{s\in\mathbb{R}^d}$ should be compared with the smooth class of states $\omega\circ\tau_s^{-1}$ introduced in \cite[Def. 5.2]{Khavkine-Melati-Moretti-17}.
	In particular $\{P_s\}$ is associated with a unique $P\in\mathrm{S}\Gamma_{\mathrm{c}}(\pi_d^*\psi^*T^*M^{\boxtimes 2})'$ -- \textit{cf.} Definition \ref{Definition: smooth compactly supported d-dimensional family of variations} in Appendix \ref{Appendix: Peetre-Slovak theorem}.
	The existence of the family $\{P_s\}_{s\in\mathbb{R}^d}$ is a consequence of the smoothness in the parameter $s\in\mathbb{R}^d$ of the elliptic operator $E_s$ associated to the background data $b_s$ -- \textit{cf.} equation 
	\eqref{Equation: definition of the elliptic operator} -- and of the construction of $P_s$ as a pseudodifferential operator \cite[Thm. 5.1]{Shubin}.
	Notice that, given $\{P_s\}_{s\in\mathbb{R}^d}$, any other family of parametrices is of the form $\{P_s+W_s\}_{s\in\mathbb{R}^d}$ where $\{W_s\}_{s\in\mathbb{R}^d}$ is a smooth family such that $W_s\in\mathrm{S}\Gamma(\psi^*TM^{\boxtimes 2})$ -- that is, $\{W_s\}_{s\in\mathbb{R}^d}$ is associated with a unique $W\in\mathrm{S}\Gamma(\pi_d^*\psi^*TM^{\boxtimes 2})$.
\end{remark}
\begin{remark}\label{Remark: on the W-dependence of the regularity property}
	Observe that, if condition \eqref{Equation: regularity assumption on Wick powers} holds true for $\Phi^\ell$ with $\ell\leq k$ then, to verify it for $\Phi^k$, it suffices to check it for any, but fixed choice of the family of $\{P_s\}_{s\in\mathbb{R}^d}$.

	The proof of this statement goes by induction: Condition \eqref{Equation: regularity assumption on Wick powers} holds true for $k=0,1$ independently from $\{P_s\}_{s\in\mathbb{R}^s}$ since both $\Phi^0[N;b]$ and $\Phi^1[N;b]=\Phi[N;b]$ identify per construction constant sections over $\mathcal{E}[N;b]$.
	Let now assume that condition \eqref{Equation: regularity assumption on Wick powers} holds true for $\mathcal{U}_{\ell}$ for all $\ell\leq k$ and for all smooth family $\{P_s\}_{s\in\mathbb{R}^d}$.
	We now show that, if \eqref{Equation: regularity assumption on Wick powers} holds true for $\mathcal{U}_{k}$ built out of a particular smooth family $\{P_s\}_{s\in\mathbb{R}^d}$, then it holds true for all distributions $\widetilde{\mathcal{U}}_k$ associated with any other smooth family $\{\widetilde{P}_s\}_{s\in\mathbb{R}^d}$.
	From the equivariance condition -- see \eqref{Equation: definition of locally covariant algebra of interest} -- it descends
	\begin{align*}
		\widetilde{\mathcal{U}}_{k}(\chi\otimes\omega)&=
		\int_{\mathbb{R}^d}\mathrm{d}s\,\Phi^k[N;b_s](\omega,\widetilde{P}_s,0)\chi(s)=
		\int_{\mathbb{R}^d}\mathrm{d}s\, \alpha^{P_s}_{\widetilde{P}_s}\big[\Phi^k[N;b_s](\omega,P_s)\big](0)\chi(s)\,.
	\end{align*}
	Definition \eqref{Equation: star-isomorphism between different fiber algebras} entails that $\alpha^{P_s}_{\widetilde{P}_s}\big[\Phi^k[N;b_s](\omega,P_s)\big]$ is a linear combination of functional derivatives with respect to $\varphi\in\Gamma(\psi^*TM)$ of $\Phi^k[N;b_s](\omega,P_s)$ evaluated at $\varphi=0$.
	By condition \eqref{Equation: inductive condition on derivative for Wick powers} each of such derivatives which is non trivial can be reduced to a Wick power $\Phi^\ell$ with $\ell\leq k$.
	Explicitly it holds
	\begin{align*}
		\alpha^{P_s}_{\widetilde{P}_s}\big[\Phi^k[N;b_s](\omega,P_s)\big](0)=
		\Phi^k[N;b_s](\omega,P_s,0)+
		\sum_{2\ell=2}^{k}\frac{k!!}{(2\ell)!!(k-2\ell)!!}\Phi^{k-2\ell}[N;b_s]([\widetilde{P}_s-P_s]^{\otimes \ell}\lrcorner\omega,P_s,0)\,,
	\end{align*}
	where locally $([\widetilde{P}_s-P_s]^{\otimes\ell}\lrcorner \omega_k)(x)=(\widetilde{P}_s-P_s)^{a_1a_2}(x,x)\dots (\widetilde{P}_s-P_s)^{a_{\ell-1} a_\ell}(x,x)(\omega_k)_{a_1\ldots a_k}(x)$.
	Observe that $\widetilde{P}_s-P_s\in\mathrm{S}\Gamma(\psi^*TM^{\boxtimes 2})$ is smooth in $s\in\mathbb{R}^d$ and moreover the last expression contains only terms of the form $\mathcal{U}_\ell$, $\ell\leq k$, where \eqref{Equation: regularity assumption on Wick powers} holds true by the inductive hypothesis.
\end{remark}

\begin{remark}\label{Remark: scaled Wick powers}
	Notice that, if $\Phi^\bullet$ identifies a family of Wick powers as per Definition \ref{Definition: Wick power associated with a locally covariant observable}, then for all $\lambda>0$ we can construct another family of Wick powers $\Phi^\bullet_\lambda$ via scaling -- \textit{cf.} Definition \ref{Definition: scaled locally covariant observable}.
	Actually we set
	\begin{equation}\label{eq: rescaled Wick}
		\Phi^k_\lambda[N;b]:=(S_\lambda\Phi^k)[N;b]:=\Phi^k[N;b_\lambda]\,,
	\end{equation}
	where $(N;b_\lambda)$ is defined in \eqref{Equation: definition of scaling in BkgG}.
	This fact will play a crucial r\^ole in Section \ref{Section: Renormalization and Ricci flow}. 
\end{remark}
\begin{Example}\label{Example: Wick powers exploited for Ricci flow}
	We provide a constructive scheme yielding a natural candidate to play the role of a family of Wick powers.
	Let $k\in\mathbb{N}$ and let $(N;b)\in\operatorname{Obj}(\mathsf{Bkg})$ and $\omega\in\mathrm{S}\Gamma_{\mathrm{c}}^{k,1}[N;b]$ be arbitrary.
	Define $\phi^k[N;b]\in\mathcal{P}_{\textrm{loc}}[N;b]$ -- see Definition \ref{Definition: smooth local polynomial functionals} -- to be 
	\begin{align}\label{Equation: k-th power polynomial functional}
		\phi^k[N;b](\omega,\varphi):=\int_\Sigma\langle\varphi^{[\otimes] k},\omega\rangle\,,
	\end{align}
	where $\varphi^{[\otimes] k}\in\Gamma(\mathrm{S}^{\otimes k}\psi^*TM)$ has been defined in remark \ref{Remark: on exterior tensor product and symmetrix sections}, while locally $\langle\varphi^{[\otimes] k},\omega\rangle=\varphi^{a_1}\dots\varphi^{a_k}\omega_{a_1\ldots a_k}$.
	Since we are interested in functionals which are equivariant with respect to the choice of $P\in\operatorname{Par}[N;b]$, we set for all $\omega\in\mathrm{S}\Gamma_{\mathrm{c}}^{k,1}[N;b]$ and $\varphi\in\Gamma(\psi^*TM)$,
	\begin{align}\label{Equation: definition of Wick powers exploited for Ricci flow}
		\wick{\Phi^k}[N;b](\omega,P,\varphi):=\bigg[\exp\big[\Upsilon_{[W_P]}]\phi^k[N;b](\omega)\bigg](\varphi)\,,
	\end{align}
	where $\Gamma_{[W_P]}$ is defined as in \eqref{Equation: Gamma contraction operator} where we also exploited the support property of the functional derivatives of any local functionals -- see Defnition \ref{Definition: smooth local polynomial functionals}.
	Moreover $[W_P]\in\Gamma(\mathrm{S}^{\otimes 2}\psi^*TM)$ is defined as in Remark \ref{Remark: local Hadamard representation of parametrices}.
	Notice that such a local decomposition depends on the chosen background geometry $(N;b)$ out of which $H$ is identified.
	
	By extending $\Phi^k[N;b]$ to a locally covariant observable -- see Remark \ref{Remark: on the alg-morphism property of locally covariant observables} -- the collection of all $\Phi^k$ defines a family of Wick powers as per Definition \ref{Definition: Wick power associated with a locally covariant observable}. Indeed observe that, adapting \eqref{eq: rescaled Wick} to $\Phi^k$, this scales almost homogeneously with dimension $\kappa=0$ while the second and third condition in Definition \ref{Definition: Wick power associated with a locally covariant observable} follow per construction. 
	Finally \eqref{Equation: regularity assumption on Wick powers} holds true since
	\begin{align*}
	\wick{\Phi^{2\ell+1}}[N;b](\omega,P,0)=0\,,\qquad
	\wick{\Phi^{2\ell}}[N;b](\omega,P,0)=\int_\Sigma\langle[W_P]^{[\otimes]\ell},\omega\rangle\,,
	\end{align*}
	where $[W_P]^{[\otimes]\ell}:=[[W_P]^{\otimes \ell}]$ -- \textit{cf.} remark \ref{Remark: on exterior tensor product and symmetrix sections}.
	The smoothness of the associated distribution $\mathcal{U}_{k}$ -- see equation \eqref{Equation: distribution associated with Wick powers} -- follows.
\end{Example}

Our next step consists of addressing the question concerning the characterization and the classification of the freedom in the construction of a family of Wick powers. In the Lorentzian setting this question has already been thoroughly investigated for a large class of field theories, see \cite{Hollands-Wald-01,Khavkine-Melati-Moretti-17,Khavkine:2014zsa}, while here we tackle the same problem for the model in hand, introduced in Section \ref{Section: Geometrical setting}.

In the same spirit of \cite[Thm. 5.2, Thm. 6.2]{Khavkine-Melati-Moretti-17} the result is divided in two parts -- see Theorems \ref{Theorem: ambiguities for Wick powers} and \ref{Theorem: structural form of ambiguities}.
In the first we prove a general formula -- Equation \eqref{Equation: ambiguities for Wick powers} -- which starts from two families of Wick powers, say $\widehat{\Phi}^\bullet$ and $\Phi^\bullet$, relating each $\widehat{\Phi}^k$ to a linear combination of $\{\Phi^\ell\}_{\ell\leq k}$ whose coefficients are a collection of locally covariant observables $\{C_\ell\}_{1\leq\ell\leq k-2}$.
This result profits of the Peetre-Slov\'ak theorem which we briefly recall in Appendix \ref{Appendix: Peetre-Slovak theorem}.
In the second part, we prove additional structural properties of the coefficients $C_\ell$, recasting in this framework \cite[Thm. 6.2]{Khavkine-Melati-Moretti-17}.

\noindent Before stating the key results of this section, we prove a key lemma.

\begin{lemma}\label{Lemma: psi,g dependence of locally covariant tensor}
	Let $k\in\mathbb{N}$ and for all $(N;b)\in\operatorname{Obj}(\mathsf{Bkg})$ let $c_k[N;b]\in\mathrm{S}\Gamma^{k,1}[N;b]$ be such that
	\begin{align*}
		c_k[\widetilde{N};\widetilde{b}]=\tau^*c_k[N,b]\,,
	\end{align*}
	for all $[\tau,t]\in\textrm{Ar}(\mathsf{Bkg})$ between $(N,b)$ and $(\widetilde{N};\widetilde{b})$.
	Then, for any $(N;b)=(\Sigma,M;\psi,\gamma,g)\in\mathrm{Obj}(\mathsf{Bkg})$, there exists a map $D_{k,\Sigma,M}\colon\Gamma(\mathrm{S}^{\otimes 2} T^*\Sigma\otimes \mathrm{S}^{\otimes 2}\psi^*T^*M)\to\Gamma^{k,1}[N;b]$ such that
	\begin{align}\label{Equation: psi,g dependence of locally covariant tensor}
	c_k[N;b]=c_k[\Sigma,M;\psi,\gamma,g]=D_{k,\Sigma,M}(\gamma,\psi^*g)\,.
	\end{align}
\end{lemma}
\begin{proof}
	Per hypothesis, for every pair  $[N;b]=(\Sigma,M;\psi,\gamma,g)$, $[\widetilde{N},\widetilde{b}]=(\widetilde{\Sigma},\widetilde{M};\widetilde{\psi},\widetilde{\gamma},\widetilde{g})\in\operatorname{Obj}(\mathsf{Bkg})$ such that there exists $(\tau,t)\in\operatorname{Ar}(\mathsf{Bkg})$ from $[N;b]$ to $[\widetilde{N},\widetilde{b}]$,  it holds
	\begin{align}\label{Equation: naturality condition exploited in the proof of the Lemma}
	c_k(\widetilde{\Sigma},\widetilde{M};\widetilde{\psi},\widetilde{\gamma},\widetilde{g})=\tau^*c_k(\Sigma,M;\psi,\gamma,g)\,.
	\end{align} 
	Consider the special case where $[\widetilde{N},\widetilde{b}]$ and $(\tau,t)$ are such that $\widetilde{\Sigma}=\Sigma$, $\widetilde{M}=M$, $\tau=\operatorname{Id}_\Sigma$ while $t\colon M\to M$ is any diffeomorphism in $M$ such that $t|_{\psi(\Sigma)}=\operatorname{Id}|_{\Sigma}$. Condition \eqref{Equation: compatibility condition for arrows in BkgG} entails $\widetilde{\psi}=\widetilde{\psi}\circ\tau=t\circ\psi=\psi$, while $g=t^*\widetilde{g}$, $\widetilde{\gamma}=\gamma$.
	Equation \eqref{Equation: naturality condition exploited in the proof of the Lemma} implies
	\begin{align*}
	c_k[N;b]=c_k(\Sigma,M;\psi,\gamma,\widetilde{g})=c_k(\Sigma,M;\psi,\gamma,t^*\widetilde{g})\,.
	\end{align*}
	It follows that $c_k(\Sigma,M;\psi,\gamma,g)$ depends on $g$ only via $\psi^*g$, that is
	\begin{align*}
	c_k(\Sigma,M;\psi,\gamma,g)=:d_k(\Sigma,M;\psi,\gamma,\psi^*g)\,.
	\end{align*}
	We prove that $c_k[N;b]$ depends on $\psi$ only via $\psi^*g$.
	As above let $\tau=\operatorname{Id}_\Sigma$ and let $t\colon M\to M$ be any diffeomorphism: by condition \eqref{Equation: psi,g dependence of locally covariant tensor} we find
	\begin{align*}
	d_k(\Sigma,M;\widetilde{\psi},\gamma,\widetilde{\psi}^*\widetilde{g})=c_k(\Sigma,M;\widetilde{\psi},\gamma,\widetilde{g})
	&=c_k(\Sigma,M;\psi,\gamma,t^*\widetilde{g})\\&
	=d_k(\Sigma,M;\psi,\gamma,\psi^*t^*\widetilde{g})
	=d_k(\Sigma,M;\psi,\gamma,\widetilde{\psi}^*\widetilde{g})\,,
	\end{align*}
	where we exploited equation \eqref{Equation: compatibility condition for arrows in BkgG} so that $\psi^*t^*\widetilde{g}=(t\circ\psi)^*\widetilde{g}=\widetilde{\psi}^*\widetilde{g}$.
	This implies that
	\begin{align*}
	c_k(\Sigma,M;\psi,\gamma,g)=:D_{k,\Sigma,M}(\gamma,\psi^*g)\,,
	\end{align*}
	which entails the sought result.
\end{proof}

\begin{theorem}\label{Theorem: ambiguities for Wick powers}
	Let $\widehat{\Phi}^\bullet$ and $\Phi^\bullet$ be two families of Wick powers associated to $\Phi$ as per Definition \ref{Definition: Wick power associated with a locally covariant observable}.
	Then, for all integers $k> 2$, there exists a collection $\{C_\ell\}_{2\leq\ell\leq k}$ of locally covariant
	observables	$C_\ell\colon\mathrm{S}\Gamma_{\mathrm{c}}^{\ell}\to\mathcal{A}$, each of which scales almost homogeneously with dimension $\kappa=0$ so that, for all $(N,b)\in\operatorname{Obj}(\mathsf{Bkg})$ and for all $\omega_1\in\mathrm{S}\Gamma_{\mathrm{c}}^{k,1}[N;b]$
	\begin{align}\label{Equation: phi-independence structure of coefficients Ck}
		C_\ell[N;b](\omega_1)=\int_\Sigma \langle c_\ell[N;b], \omega_1\rangle \;1_{\mathcal{A}[N;b]}\,,
	\end{align}
	where $c_\ell[N;b]\in\mathrm{S}\Gamma^{\ell,1}[N;b]$ for all $(N;b)\in\mathrm{Obj}(\mathsf{Bkg})$.
	Furthermore, if $(N;b)=(\Sigma,M;\psi,\gamma,g)$, then
	\begin{align}\label{Equation: differential structure of tensors coefficients}
		c_\ell(\Sigma,M;\psi,\gamma,g)=D_{\ell,\Sigma,M}(\gamma,\psi^*g)\,,
	\end{align}
	where $D_{\ell,\Sigma,M}\colon\Gamma(\mathrm{S}^{\otimes 2}T^*\Sigma\otimes \mathrm{S}^{\otimes 2}\psi^*T^*M)\to\Gamma(\mathrm{S}^{\otimes\ell}\psi^*TM)$ is a differential operator of locally bounded order in the sense of Definition \ref{Definition: differential operator} in Appendix \ref{Appendix: Peetre-Slovak theorem}.
	In addition, for all $(N;b)\in\operatorname{Obj}(\mathsf{Bkg})$ and for all $\omega_1\in\mathrm{S}\Gamma_{\mathrm{c}}^{k,1}[N;b]$,
	\begin{align}\label{Equation: ambiguities for Wick powers}
		\widehat{\Phi}^k[N;b](\omega_1)=\Phi^k[N;b](\omega_1)+\sum_{\ell=0}^{k-2}{k\choose \ell}\Phi^\ell[N;b](c_{k-\ell}[N;b]\lrcorner\omega_1)\,,
	\end{align}
	where $c_{k-\ell}[N;b]\lrcorner\omega_1\in\mathrm{S}\Gamma_{\mathrm{c}}^{\ell,1}[N;b]$ reads locally $\big(c_{k-\ell}[N;b]\lrcorner\omega_1\big)_{a_1\ldots a_\ell}=c_{k-\ell}^{a_{\ell+1}\ldots a_{k}}[N;b](\omega_1)_{a_1\ldots a_k}$.
\end{theorem}

\begin{proof}
	The proof proceeds per induction with respect to $k$.
	First of all we prove Equation \eqref{Equation: ambiguities for Wick powers} for $k=2$.
	Hence, we set $C_2:=\widehat{\Phi}^2-\Phi^2$, showing that it is of the wanted form. Let  $(N;b)\in\operatorname{Obj}(\mathsf{Bkg})$ and let $\omega_1\in\mathrm{S}\Gamma_{\mathrm{c}}^{2,1}[N;b]$, while  $\varphi_1,\varphi_2\in\Gamma(\psi^*TM)$ and $P\in\operatorname{Par}[N;b]$.
	Equation \eqref{Equation: inductive condition on derivative for Wick powers} entails
	\begin{align*}
		\big\langle C_2[N;b](\omega_1,P)^{(1)}[\varphi_1],\varphi_2\big\rangle=2(\Phi-\Phi)[N;b](\varphi_2\lrcorner\omega_1,P,\varphi_1)=0\,.
	\end{align*}
It follows that, as an element of $\mathcal{A}[N;b]$, $C_2[N;b](\omega_1)$ does not depend on $(P,\varphi)$, that is, it is a multiple of the identity element:
	\begin{align*}
		C_2[N;b](\omega_1,P,\varphi)=\int_\Sigma \langle c_2[N;b],\omega_1\rangle\; 1_{\mathcal{A}[N;b]}\,,
	\end{align*}
	where $c_2$ is an assignment to $(N;b)\in\mathrm{Obj}(\mathsf{Bkg})$ of an element in $\mathrm{S}\Gamma^{2,1}[N;b]$ on account of the regularity condition \eqref{Equation: regularity assumption on Wick powers}.
	Moreover $C_2$ inherits from $\Phi^2$ and $\widehat{\Phi}^2$ the property of scaling almost homogeneously with degree $\kappa=0$. 
	Since the arrows of $\mathsf{Bkg}$ act on $c_2[N;b]$ via pull-back, the hypotheses of Lemma \ref{Lemma: psi,g dependence of locally covariant tensor} are met and we can conclude that
	$c_2(\Sigma,M;\psi,\gamma,g)=D_{2,\Sigma,M}(\gamma,\psi^*g)$
	for all $(N;b)=(\Sigma,M;\psi,\gamma,g)$.
It descends that, for all $x\in\Sigma$, $D_{2,\Sigma,M}(\gamma,\psi^*g)(x)$ depends only on the germ of $\gamma,\psi^*g$ at $x$.
	Furthermore, condition \eqref{Equation: regularity assumption on Wick powers} ensures that $(\gamma,\psi^*g)\mapsto D_{2,\Sigma,M}(\gamma,\psi^*g)$ is weakly regular as per Definition \ref{Definition: weak regularity condition}.
	By the Peetre-Slov\'ak Theorem -- see Appendix \ref{Appendix: Peetre-Slovak theorem} -- it follows that $D_{2,\Sigma,M}\colon\Gamma(\mathrm{S}^{\otimes 2} T^*\Sigma\otimes \mathrm{S}^{\otimes 2}\psi^*T^*M)\to\Gamma(\mathrm{S}^{\otimes 2}\psi^*TM)$ is a differential operator of locally bounded order. This concludes the proof of the theorem for $k=2$.	
	
	Let us assume that, for all $2\leq p\leq k$, $(N,b)\in\operatorname{Obj}(\mathsf{Bkg})$ and for all $\omega_1\in\mathrm{S}\Gamma_{\mathrm{c}}^{p,1}[N;b]$
	\begin{align}\label{Equation: inductive hypothesis}
		\widehat{\Phi}^p[N;b](\omega_1)=\Phi^p[N;b](\omega_1)+\sum_{\ell=0}^{p-2}{p\choose \ell}\Phi^\ell[N;b](c_{p-\ell}[N;b]\lrcorner\omega_1)\,.
	\end{align}
	Here for all $q\in\{1,\ldots,p-2\}$, $C_q$ is a locally covariant observable which scales almost homogeneously with dimension $\kappa=0$, so that
	\begin{align*}
		C_q[N;b](\omega_1)=\int_\Sigma\langle c_q[N;b],\omega_1\rangle\;1_{\mathcal{A}[N;b]}\qquad\forall \omega_1\in\mathrm{S}\Gamma_{\mathrm{c}}^{q,1}[N;b]\,,
	\end{align*}
	where $c_q[N;b]=c_q(\Sigma,M;\psi,\gamma,g)=D_{q,\Sigma,M}(\gamma,\psi^*g)\in\mathrm{S}\Gamma^{q,1}[N;b]$, being $D_{q,\Sigma,M}$ a differential operator of locally bounded order.
	We prove the inductive step, namely that equation \eqref{Equation: inductive hypothesis} holds true for $p=k+1$. As for the case $k=2$, let $C_{k+1}$ be defined as
	\begin{align}
		C_{k+1}[N;b](\omega_1):=\widehat{\Phi}^{k+1}[N;b](\omega_1)-\Phi^{k+1}[N;b](\omega_1)-\sum_{\ell=0}^{k-1}{k+1\choose \ell}\Phi^\ell[N;b](c_{k+1-\ell}[N;b]\lrcorner\omega_1)\,,
	\end{align}
	for all $\omega_1\in\mathrm{S}\Gamma_{\mathrm{c}}^{k+1,1}[N;b]$.
	Equation \eqref{Equation: inductive condition on derivative for Wick powers} and the inductive hypothesis \eqref{Equation: inductive hypothesis} entail that $C_{k+1}[N;b](\omega_1)$ is an element of $\mathcal{A}[N;b]$ that does not depend on the choice of $(P,\varphi)$.
	Hence there exist an assignment to $(N;b)\in\mathrm{Obj}(\mathsf{Bkg})$ of an element $c_{k+1}[N;b]\in\mathrm{S}\Gamma^{k+1,1}[N;b]$ such that
	\begin{align*}
		C_{k+1}[N;b](\omega_1)=\int_\Sigma\big\langle c_{k+1}[N;b],\omega_1\big\rangle\;1_{\mathcal{A}[N;b]}\,.
	\end{align*}
	where we used the regularity condition \eqref{Equation: regularity assumption on Wick powers}. In addition, still on account of the inductive hypothesis \eqref{Equation: inductive hypothesis}, $C_{k+1}$ scales almost homogeneously with degree $\kappa=0$.
	This implies that $c_{k+1}[N;b]$ satisfies the hypothesis of Lemma \ref{Lemma: psi,g dependence of locally covariant tensor} and, thus, it follows that $c_{k+1}(\Sigma,M;\psi,\gamma,g)=D_{k+1,\Sigma,M}(\gamma,\psi^*g)$.
	The regularity condition \eqref{Equation: regularity assumption on Wick powers} ensures that $D_{k+1,\Sigma,M}\colon\Gamma(\mathrm{S}^{\otimes 2} T^*\Sigma\otimes \mathrm{S}^{\otimes 2} \psi^*T^*M)\to\Gamma(S^{\otimes k+1}\psi^*TM)$ is weakly regular and that, for all $x\in\Sigma$, $D_{k+1,\Sigma,M}(\gamma,\psi^*g)(x)$ depends on $\gamma,\psi^*g$ only via their germs at $x$.
	By the Peetre-Slov\'ak Theorem $D_{k+1,\Sigma,M}$ is a differential operator of locally bounded order. This completes the proof.
\end{proof}

\noindent To conclude we state the last result of this section.

\begin{theorem}\label{Theorem: structural form of ambiguities}
	Under the same assumptions of Theorem \ref{Theorem: ambiguities for Wick powers}, it holds that, for each $k\in\mathbb{N}$ and for each $(N;b)=(\Sigma,M;\psi,\gamma,g)\in\operatorname{Obj}(\mathsf{Bkg})$, the map $D_{
	k,\Sigma,M}$ defined in equation \eqref{Equation: differential structure of tensors coefficients} enjoys the following properties:
	\begin{enumerate}
		\item
		$D_{k,\Sigma,M}\colon\Gamma(\mathrm{S}^{\otimes 2} T^*\Sigma\otimes \mathrm{S}^{\otimes 2}\psi^*TM)\to\Gamma^{k,1}[N;b]$ is a differential operator of globally bounded order -- see Definition \ref{Definition: differential operator} of Appendix \ref{Appendix: Peetre-Slovak theorem} ;
		\item
		for all $x\in\Sigma$, $\gamma\in\Gamma(\mathrm{S}^{\otimes 2} T^*\Sigma)$ and $\psi^*g\in\Gamma(\mathrm{S}^{\otimes 2} \psi^*T^*M)$ it holds
		\begin{align}
		\nonumber
		D_{k,\Sigma,M}(\gamma,\psi^*g)(x)= \mathsf{D}_k\bigg(
		\gamma^{\alpha\beta}(x),\epsilon^{\alpha\beta}(x), R_{\alpha\beta\mu\nu}[\gamma](x),\ldots\nabla^\Sigma_{\alpha_1}\dots\nabla^\Sigma_{\alpha_p}R_{\alpha\beta\mu\nu}[\gamma](x),\ldots,\\
		g^{ab}(\psi(x)),R_{abcd}[g](\psi(x)),\ldots,\nabla^M_{a_1}\dots\nabla^M_{a_r}R_{abcd}[g](\psi(x))
		\bigg)\,,\label{eq:Dk}
		\end{align}
		where $\mathsf{D}_k$ is a tensor, covariantly constructed from its arguments, where the symbol $R$ in the above expression indicates the Riemann tensor while $\epsilon^{\alpha\beta}$ is the totally antisymmetric Levi-Civita tensor.
		\item 
		Each $\mathsf{D}_k$ is an homogeneous of degree $\kappa=0$, linear combination of finitely many covariantly constructed tensors. These are polynomials in all the arguments on which $\mathsf{D}_k$ depends in \eqref{eq:Dk} and the functional form does not depend on the choice of $[N;b]\in\mathrm{Obj}(\mathsf{Bkg})$.		
	\end{enumerate}
\end{theorem}
\begin{proof}

On account of Lemma \ref{Lemma: psi,g dependence of locally covariant tensor} the proof can follow almost slavishly that of \cite[Thm. 6.2]{Khavkine-Melati-Moretti-17}. For this reason we omit it.
\end{proof}

\section{Renormalization and Ricci flow}\label{Section: Renormalization and Ricci flow}
Our main goal is to apply the results of Section \ref{Section: Wick powers and ambiguities} giving a rigorous derivation of the Ricci flow from the renormalization of (the linearisation of) the non-linear Sigma-model introduced in Section \ref{Section: Geometrical setting} -- see \cite{Carfora-17} and also \cite{Gaw99}.

\paragraph{Perturbative Euclidean statistical field theory.}
In the framework of Euclidean algebraic quantum field theory the expectation value of a (locally covariant) observable $\mathcal{O}$ is typically built out of a Lagrangian density $\mathcal{L}$, as the one introduced in equation \eqref{Equation: second order expansion of harmonic Lagrangean}, which is regarded as the covariance of an infinite dimensional Gaussian measure. However, except for some rather special cases, this approach brings several difficulties in dealing with non-linearities and thus one must resort to a perturbative approach. Fixing $\mathcal{L}$ to be the one of \eqref{Equation: second order expansion of harmonic Lagrangean}, following the discussion and the notation at the beginning of Section \ref{Section: Geometrical setting}, we split $\mathcal{L}$ in two contributions
\begin{align}
	\label{Equation: split of Lagrangean density}
	\mathcal{L}(\psi,\gamma,g;\varphi)&:=
	\mathcal{L}_{\mathrm{free}}(\psi,\gamma,g;\varphi)+\mathcal{L}_{\mathrm{int}}(\psi,\gamma,g;\varphi)\,,\\
	\label{Equation: free part of the Lagrangean density}
	\mathcal{L}_{\mathrm{free}}(\psi,\gamma,g;\varphi)&:=-\frac{\nu^2}{2}\langle\varphi,E\varphi\rangle\mu_\gamma\,,\\
	\label{Equation: interacting part of the Lagrangean density}
	\mathcal{L}_{\mathrm{int}}(\psi,\gamma,g,\varphi)&:=\mathcal{L}_{\mathrm{H}}(\psi,\gamma,g)+
	\bigg[
	\nu g(\varphi,Q(\psi))+
	\frac{\nu^2}{2}h(\mathrm{Riem}(\varphi,\mathrm{d}\psi)\varphi,\mathrm{d}\psi)
	\bigg]\mu_\gamma\,.
\end{align}
As the notation suggests, we interpret $\mathcal{L}_{\mathrm{free}}$ as the Lagrangian density of a free field theory while $\mathcal{L}_{\mathrm{int}}$ is interpreted as an interacting part. 

One has to keep in mind that such subdivision is arbitrary and our choice is dictated by the fact that the dynamics encoded in $\mathcal{L}_{\mathrm{free}}$ is ruled by the elliptic operator $E$. According to Proposition \ref{Prop: ELCFT}, we can associate to it an Euclidean locally covariant theory $\mathcal{A}:\mathsf{Bkg}\to\mathsf{Alg}$. At the same time, to $\mathcal{L}_{\mathrm{int}}$ we can associate a locally covariant observable as per Definition \ref{Def:nat_transf} with the following procedure.

Consider a family of Wick powers $\Phi^\bullet$ as per Definition \ref{Definition: Wick power associated with a locally covariant observable} and, starting from $\mathcal{L}_{\mathrm{int}}$, define the following natural transformation which we indicate for simplicity as $\mathcal{L}_{\mathrm{int}}[\Phi^\bullet]:C^\infty_{\mathrm{c}}\to\mathcal{A}$:
\begin{align}\label{Equation: locally covariant interacting Lagrangean density for a fixed family of Wick powers}
	\mathcal{L}_{\mathrm{int}}[\Phi^\bullet][N;b](f):=
	\mathcal{L}_{\mathrm{H}}[N;b](f)\;1_{\mathcal{A}[N;b]}+
	\nu\Phi[N;b](fQ(\psi)\mu_\gamma)+
	\frac{\nu^2}{2}\Phi^2[N;b](f\theta[N;b]\mu_\gamma)\,,
\end{align}
where $(N;b)\in\operatorname{Obj}(\mathsf{Bkg})$, $f\in C^\infty_{\mathrm{c}}(\Sigma)$.
The covariant functor $C^\infty_{\mathrm{c}}\colon\mathsf{Bkg}\to\mathsf{Alg}$ has been defined in definition \ref{Definition: action of the functor Gammak on Obj and Arr}.
In addition $\theta[N;b]\in\Gamma(\mathrm{S}^{\otimes 2}\psi^*TM)$ is locally defined by
\begin{align}\label{Equation: convenient tensor}
	\theta[N;b]_{cd}(x):=\gamma^{\alpha\beta}(x)g_{\ell b}(\psi(x))R_{cad}^{\phantom{cad}{\ell}}[g](\mathrm{d}\psi)^a_{\alpha}(\mathrm{d}\psi)^b_{\beta}\,, 
\end{align}
while $\mathcal{L}_{\mathrm{H}}[N;b](f):=\int_\Sigma d\mu_\gamma\, f\mathcal{L}_{\mathrm{H}}(\psi,\gamma,g)$, $\mathcal{L}_{\mathrm{H}}$ being the Lagrangian in \eqref{Equation: harmonic Lagrangean density}.

Within the perturbative approach to Euclidean field theory one defines the (generating function) partition function as the natural transformation $\mathcal{Z}[\Phi^\bullet]:C^\infty_{\mathrm{c}}\to\mathcal{A}$ 
\begin{align}\label{Equation: partition function}
	\mathcal{Z}[\Phi^\bullet][N;b](f):=\exp_{\mathcal{A}}\bigg[z\mathcal{L}_{\mathrm{int}}[\Phi^k][N;b](f)\bigg]:=
	\sum_{n\geq 0}\frac{z^n}{n!}\mathcal{L}_{\mathrm{int}}[\Phi^k][N;b](f)^n\in\mathcal{A}[N;b][[z]]\,,
\end{align}
where the exponential series is considered as a formal power series in the formal parameter $z$ and the product is the one defined by $\mathcal{A}[N;b]$, for all $[N;b]\in\mathrm{Obj}(\mathsf{Bkg})$.
Out of $\mathcal{Z}[\Phi^\bullet]$ one can build pertubatively the above mentioned expectation value of any locally covariant observable $\mathcal{O}$.
A complete discussion of the structural properties of the perturbative approach to Euclidean algebraic field theories is beyond the scope of this paper and we postpone it to a forthcoming work \cite{Dappiaggi-Drago-Rinaldi-19}.

\paragraph{Application to Ricci flow.}
As explained at the beginning of this section, once it has been fixed
a family of Wick powers $\Phi^\bullet$
-- {\it cf.} Definition \ref{Definition: Wick power associated with a locally covariant observable} -- we may define a corresponding locally covariant Lagrangian density $\mathcal{L}_{\mathrm{int}}[\Phi^\bullet]$ -- see equation \eqref{Equation: locally covariant interacting Lagrangean density for a fixed family of Wick powers} -- and the associated partition function $\mathcal{Z}[\Phi^\bullet]$ as per equation \eqref{Equation: partition function}.
The key point consists in realizing that different choices of $\Phi^\bullet$ yield different explicit forms for $\mathcal{L}_{\mathrm{int}}[\Phi^\bullet]$, which, on account of theorems \ref{Theorem: ambiguities for Wick powers} and \ref{Theorem: structural form of ambiguities},  differ only by a linear combination of terms proportional to certain locally covariant quantum fields -- see equation \ref{Equation: ambiguities for Wick powers}.

In the framework of the renormalization group approach such ambiguity is studied by choosing, for each real $\lambda>0$, $\widehat{\Phi}^\bullet:=S_\lambda\Phi^\bullet$, see Remark \ref{Remark: scaled Wick powers}, in particular Equation \eqref{eq: rescaled Wick}.
As a consequence we consider an interacting Lagrangean density $\mathcal{L}_{\mathrm{int}}[S_\lambda\Phi^\bullet]$ which, by Theorem \ref{Theorem: ambiguities for Wick powers} can be written as
$\mathcal{L}_{\mathrm{int}}[S_\lambda\Phi^\bullet]=\mathcal{L}_{\mathrm{int}}[\Phi^\bullet]+\mathcal{R}_\lambda[\Phi^\bullet]\,,$
where $\mathcal{R}_\lambda[\Phi^\bullet]$ is a suitable remainder. The main idea behind the renormalization group approach is that $\mathcal{R}_\lambda[\Phi^\bullet]$ can be reabsorbed in the full Lagrangian density, namely, for every $\lambda>0$, there exists a natural transformation, dubbed renormalized Lagrangian at the scale $\lambda$, $\mathcal{L}_{{\mathrm{int}},\lambda}[\Phi^\bullet]:C^\infty_{\mathrm{c}}\to\mathcal{A}$ such that
\begin{align}\label{Equation: renormalized interacting Lagrangean density}
	\mathcal{L}_{\mathrm{int}}[S_\lambda\Phi^\bullet]
	=\mathcal{L}_{\mathrm{int}}[\Phi^\bullet]+\mathcal{R}_\lambda[\Phi^\bullet]	
	=:\mathcal{L}_{\mathrm{int},\lambda}[\Phi^\bullet]\,.
\end{align}

In what follows we will compute explicitly the renormalized Lagrangian density at scale $\lambda>0$, $\mathcal{L}_\lambda[\Phi^\bullet]:=\mathcal{L}_{\mathrm{free}}[\Phi^\bullet]+\mathcal{L}_{\mathrm{int},\lambda}[\Phi^\bullet]$ -- see Theorem \ref{Theorem: renormalized Lagrangean density for Ricci flow application}.
For concreteness we will work with the family of Wick powers defined in Example \ref{Example: Wick powers exploited for Ricci flow}, though any well-defined, different choice can be made without affecting the final result.
Eventually we comment how the result of Theorem \ref{Theorem: renormalized Lagrangean density for Ricci flow application} are linked to the derivation of the Ricci flow \cite{Carfora-17} -- see Lemma \ref{lemma: derivation of the Ricci flow equation}.

\begin{theorem}\label{Theorem: renormalized Lagrangean density for Ricci flow application}
	Let $\wick{\Phi^\bullet}$
	be the family of Wick powers	as per Example \ref{Example: Wick powers exploited for Ricci flow}
	and let $\mathcal{L}_{\mathrm{int}}[\Phi^\bullet]$ be the locally covariant interacting Lagrangian density as per Equation \eqref{Equation: locally covariant interacting Lagrangean density for a fixed family of Wick powers}.
	For all $\lambda>0$, let $\mathcal{L}_{\mathrm{int}}[S_\lambda\Phi^\bullet]$ be the counterpart of $\mathcal{L}_{\mathrm{int}}[\Phi^\bullet]$ -- defined in \eqref{Equation: renormalized interacting Lagrangean density} -- constructed out of the rescaled natural transformation  $S_\lambda\Phi^\bullet$.
	Then it holds
	\begin{align}
	\mathcal{L}_{\mathrm{int},\lambda}[\Phi^\bullet][N;b](f)=
	\mathcal{L}_{\mathrm{H},\lambda}[N;b](f)\;1_{\mathcal{A}[N;b]}+
	\nu\Phi[N;b](fQ(\psi)\mu_\gamma)+
	\frac{\nu^2}{2}\Phi^2[N;b](f\theta[N;b]\mu_\gamma)\,,
	\end{align}
	where $f\in C^\infty_{\mathrm{c}}(\Sigma)$ and $[N;b]=(\Sigma,M;\psi,\gamma,g)$ while  
	\begin{equation}\label{eq:metric_rescaled}
	\mathcal{L}_{\mathrm{H},\lambda}[N;b](f):=\int_\Sigma f\operatorname{tr}_\gamma(\psi^*g_{\log\lambda})\mu_\gamma\,,\qquad
	(g_{\log\lambda})_{ab}(x)=g_{ab}(x)-\nu^2\log(\lambda) R_{ab}[g](x)\,.
	\end{equation}
\end{theorem}

\begin{proof}
	Let $[N;b]=(\Sigma,M;\psi,\gamma,g)\in\operatorname{Obj}(\mathsf{Bkg})$, $f\in C^\infty_{\mathrm{c}}(\Sigma,M)$, $P\in\operatorname{Par}[N;b]$ and $\varphi\in\Gamma(\psi^*TM)$.
	Recalling Definition \ref{Definition: scaled locally covariant observable} as well as equation \eqref{Equation: definition of Wick powers exploited for Ricci flow} -- see Example \ref{Example: Wick powers exploited for Ricci flow} -- it holds
	\begin{align}\label{Equation: useful equation into proof of main theorem 1}
	\mathcal{L}_{\mathrm{int}}[S_\lambda\Phi^\bullet][N;b]=\mathcal{L}_\mathrm{int}[\Phi^\bullet][N;b_\lambda]\,.
	\end{align}

	Recalling Proposition \ref{Proposition: coincinding point limit of scaled and non-scaled Hadamard parametrix} and Remark \ref{Remark: different local representation of a parametrix under scaling}, we can use the local Hadamard representation of the parametrix $P=H+W$ for $E$ to realize that, in a geodesic neighbourhood of any point $x\in\Sigma$, $P=H_\lambda+W_\lambda$ where $[W_{P,\lambda}]^{ab}(x):=W^{ab}_\lambda(x,x)=[W_P]^{ab}(x)-2\log(\lambda)g^{ab}(\psi(x))$.
	
	Using \eqref{Equation: definition of Wick powers exploited for Ricci flow}, the first two terms in $\mathcal{L}_{\mathrm{int}}$ in \eqref{Equation: locally covariant interacting Lagrangean density for a fixed family of Wick powers} remain unchanged because they are respectively constant and linear in $\varphi$. On the contrary, the third term yields
	\begin{align}
	\nonumber
	\frac{\nu^2}{2}\Phi^2[N;b_\lambda][f\,\theta[N;b_\lambda]\mu_{\gamma_{\lambda}},P,\varphi] &=
	\frac{\nu^2}{2}\Phi^2[N;b][f\,\theta[N;b]\mu_\gamma,P,\varphi]\\
	\label{Equation: useful equation into proof of main theorem 2}
	&-
	\nu^2\log(\lambda)\int_\Sigma\gamma^{\alpha\beta}R_{ab}[g](\mathrm{d}\psi)^a_{\alpha}(\mathrm{d}\psi)^b_{\beta}f\mu_\gamma\,,
	\end{align}
	where we used \eqref{Equation: convenient tensor} for $\theta[N;b]$. Inserting Equation \eqref{Equation: useful equation into proof of main theorem 2} in \eqref{Equation: useful equation into proof of main theorem 1}, the sought result descends.
\end{proof}
\begin{remark}\label{Remark: on dependence of Ricci flow w.r.t. Wick power family}
	Notice that the results of Theorem \ref{Theorem: renormalized Lagrangean density for Ricci flow application} depend on the particular choice for $\wick{\Phi^\bullet}$.
	As a matter of fact Theorem \ref{Theorem: ambiguities for Wick powers} entails that any other choice, say $\Phi^\bullet$ would be so that $\Phi^2=\wick{\Phi^2}+C_2$, being $C_2$ a locally covariant quantum field which scales almost homogeneously with degree zero -- \textit{cf.} Theorem \ref{Theorem: structural form of ambiguities}.
	As a by product, Equation \eqref{Equation: useful equation into proof of main theorem 2} holds true also for $\Phi^2$ if and only if $C_2$ scales \textit{exactly} homogeneously with degree zero -- that is, if and only if $C_2$ is invariant under scaling.
	Nevertheless, theorems \ref{Theorem: ambiguities for Wick powers}-\ref{Theorem: structural form of ambiguities} allow us to control the variation of equation \eqref{Equation: useful equation into proof of main theorem 2} with respect to the chosen family of Wick powers $\Phi^\bullet$.
\end{remark}
\begin{lemma}[{\em Ricci flow}]\label{lemma: derivation of the Ricci flow equation}
	Under the assumptions of Theorem \ref{Theorem: renormalized Lagrangean density for Ricci flow application}, setting $\lambda:=e^{2\tau}$, the corresponding metric $g(\tau):=g_{2\tau}$ as per Equation \eqref{eq:metric_rescaled} satisfies
	\begin{align}\label{Equation: Ricci flow equation}
	\frac{\mathrm{d}}{\mathrm{d}\tau}g(\tau)=-2\nu^2\mathrm{Ric}[g]=-2\nu^2\mathrm{Ric}[g(\tau)]+O(\nu^3)\,.	
	\end{align}
\end{lemma}	
\begin{proof}	
	Considering Equation \eqref{eq:metric_rescaled} and recalling the approximation made in Section \ref{Section: Geometrical setting},
	\begin{align*}
	g^{ab}_{\log\lambda}= g^{ab}-\nu^2\log(\lambda) R^{ab}[g]+O(\nu^3)\,,\qquad
	\nu^2\log(\lambda)R_{ab}[g_{\log\lambda}]
	=\nu^2\log(\lambda)R_{ab}[g]+O(\nu^3)\,.
	\end{align*}	
	Neglecting $O(\nu^3)$-contributions the previous equation leads to the wanted Ricci flow equation for the renormalized metric $g(\tau)$.
\end{proof}

\begin{remark}\label{Remark: higher order corrections to Ricci flow}
	It appears clear that the above derivation of the Ricci flow equation \eqref{Equation: Ricci flow equation} is linked to the expansion in powers of $\nu$ made in the previous Section \ref{Section: Geometrical setting}.
	It is also possible to consider an higher order expansion for the Lagrangian density \eqref{Equation: second order expansion of harmonic Lagrangean}, which leads to a corresponding improved Ricci flow equation.
	As an example, the expansion up to $o(\nu^4)$ leads to the so-called Ricci flow equation at two loops, also known as $RG-2$ flow \cite{Carfora-Guenther-18}.
	It is noteworthy that, thanks to Theorem \ref{Theorem: ambiguities for Wick powers}, the present framework can be used to obtain an analogous of Theorem \ref{Theorem: renormalized Lagrangean density for Ricci flow application} from which the higher order corrections to the Ricci flow can be explicitly computed. Yet, in this work, we refrain from providing a detailed computation, which follows the lines of the proof of Theorem  \ref{Theorem: renormalized Lagrangean density for Ricci flow application}.
\end{remark}
\begin{remark}\label{Remark: comments on the off-shell psi-independence}
	We stress that, in our derivation of the Ricci flow equation -- see Remark \ref{lemma: derivation of the Ricci flow equation} -- as well as in the proof of all the results of the previous Sections, we only assume that $\psi\in C^\infty(\Sigma;M)$.
	In particular, we do not require $\psi$ to be harmonic and the results of Theorem \ref{Theorem: renormalized Lagrangean density for Ricci flow application} and Remark \ref{lemma: derivation of the Ricci flow equation} do not depend on $\psi$.
	Stated differently, the results of this paper hold true also considering \textit{off-shell} background configuration $\psi$, rather that on-shell (harmonic) background configurations.
\end{remark}

\appendix

\section{Hadamard expansion for the parametrix of $E$}\label{Appendix: Hadamard expansion for the parametrix of E}

Goal of this appendix is to give a finer description of the local structure of a parametrix $P$ associated with the elliptic operator $E$, introduced in Equation \eqref{Equation: definition of the elliptic operator}. Let $[N;b]=(\Sigma,M;\psi,\gamma,g)\in\mathrm{Obj}(\mathsf{Bkg})$ be arbitrary but fixed. In the following, we will be considering convex, geodesic neighbourhoods of $\Sigma$, but at the same time we will be concerned about their image under the action of $\psi$ which is smooth, but not necessarily proper. Hence, whenever we consider convex, geodesic neighbourhoods of a point, we are implicitly constructing them as follows: For any $x\in\Sigma$, consider $\psi(x)\in M$ and any convex, geodesic neighbourhood $U\subset M$ centred at this point. Being $\psi$ smooth, $\psi^{-1}(U)$ is an open subset of $\Sigma$ centred at $x$. If this is not a convex, geodesic neighbourhood, then consider an open subset, which we identify with $\mathcal{O}$, which has this property. In addition $\psi(\mathcal{O})$ is a subset of $U$ and hence any two points therein are connected by a unique geodesic of $(M,g)$.

\vskip .2cm

\noindent We summarize our results in the following proposition:

\begin{proposition}\label{Proposition: coincinding point limit of scaled and non-scaled Hadamard parametrix}
	Let $(N;b)=(\Sigma,M,\gamma,g,\psi)\in\operatorname{Obj}(\mathsf{Bkg})$ and let $E\colon\Gamma(\psi^*TM)\to\Gamma(\psi^*T^*M)$ be the elliptic operator defined in \eqref{Equation: definition of the elliptic operator}.
	For $\lambda>0$ let $H, H_\lambda\in\mathrm{S}\Gamma_{\mathrm{c}}(\psi^*T^*M^{\boxtimes 2})'$ be the Hadamard parametrices associated with background data $(N;b)$ and $(N;b_\lambda)$ respectively -- \textit{cf.} Remark \ref{Remark: local Hadamard representation of parametrices} and Definition \ref{Remark: definition of scaling in BkgG}. It holds
	\begin{align}\label{Equation: coinciding point limit of scaled and non-scaled Hadamard parametrix}
		H^{bc}_\lambda(x)-H^{bc}(x)=-2\log(\lambda)V^{bc}(x)\,,
	\end{align}
	where $V\in\Gamma(\mathrm{S}^{\otimes 2}\psi^*T\mathcal{O})$ is constructed out the background data $(\psi,\gamma,g)$ -- \textit{cf.} equation \eqref{Equation: hierarchy of equation for Vns} -- and it satisfies $[V]^{bc}(x):=g^{bc}(\psi(x))$.
\end{proposition}
\begin{proof}
	Let $\mathcal{O}$ be a geodesically convex neighbourhood of $\Sigma$.
	We begin by recalling the construction of the so-called Hadamard parametrix associated to the restriction to $\mathcal{O}$ of $E$ on the background data $(N;b)$.
	This is defined as the bi-distribution $H\in\mathrm{S}\Gamma_{\mathrm{c}}(\psi^*T^*\mathcal{O}^{\boxtimes 2})'$ whose integral kernel reads \cite{Garabedian-98,Moretti-99a,Moretti-99b}
	\begin{align}\label{Equation: ansaltz for Hadamard parametrix}
		H^{bc}(x,x^\prime):=V^{bc}(x,x^\prime)\log\frac{\sigma(x,x^\prime)}{\ell_H^2}:=\sum_{n\geq 0}V_n^{bc}(x,x^\prime)\sigma(x,x^\prime)^n\log\frac{\sigma(x,x^\prime)}{\ell_H^2}\,,
	\end{align}
	where $\sigma(x,x^\prime)$ denotes the halved squared geodesic distance between $x,x^\prime\in\mathcal{O}$, while $\ell_H\in\mathbb{R}$ is an arbitrary reference length, which will play no r\^ole in the proof.
	Before focusing on the tensor coefficients $V_n^{bc}(x,x^\prime)$, observe that Equation \eqref{Equation: ansaltz for Hadamard parametrix} defines $H$ in terms of the so-called Hadamard expansion which is a formal power series in $\sigma$. Hence, with a slight abuse of notation, we have left implicit both the existence of a suitable cut-off which ensures convergence of \eqref{Equation: ansaltz for Hadamard parametrix} and the necessity or replacing $\sigma$ with a regularized counterpart $\sigma+i\epsilon$, which controls the singularity as $x=x^\prime$. Neither the cut-off nor the regularization will play a r\^ole in our analysis.

	We focus now on the remaining unknowns, the tensor coefficients $V_n^{bc}$ of \eqref{Equation: ansaltz for Hadamard parametrix}. Recalling that both $HE$ and $EH$ coincide with the identity operator up to smooth terms, it holds locally that
	\begin{align}
		\nonumber
		(EH)^c_a&=\sum_{n\geq 0}E(V_n)^c_{\phantom{c}a}\sigma^n\log\frac{\sigma}{\ell_H^2}\\ \nonumber
		&+\sum_{n\geq 0}\bigg[
		ng_{ab}V_n^{bc}\big(\Delta_\gamma\sigma+2(n-1)\big)+
		2ng_{ab}\gamma^{\alpha\beta}(\nabla^\psi V_n)^{bc}_{\phantom{bc}\alpha}(\mathrm{d}\sigma)_\beta
		\bigg]\sigma^{n-1}\log\frac{\sigma}{\ell^2_H}\\
		\label{Equation: explicit form of EH}
		&+\sum_{n\geq 0}\bigg[
		2g_{ab}\gamma^{\alpha\beta}(\nabla^\psi V_n)^{bc}_{\phantom{bc}\alpha}(\mathrm{d}\sigma)_\beta+g_{ab}V_n^{bc}\big(\Delta_\gamma\sigma-2+4n\big)
		\bigg]\sigma^{n-1}\,,
	\end{align}
	where we omitted for notational simplicity the explicit dependence on $(x,x^\prime)$ and where we exploited the identity $\gamma^{\alpha\beta}(\mathrm{d}\sigma)_\alpha(\mathrm{d}\sigma)_\beta=2\sigma$, see {\it e.g.} \cite{Poisson:2011nh}.
	To ensure that $EH-\operatorname{Id}_{\Gamma_{\mathrm{c}}(\psi^*T^*M)}\in\Gamma(\psi^*TM\boxtimes\psi^*T^*M)$, the coefficients multiplying $\log\sigma$ and $\sigma^{-1}$ ought to vanish. This leads to the following hierarchy of equations for $V_n^{bc}$:
	\begin{subequations}\label{Equation: hierarchy of equation for Vns}
		\begin{align}
		\label{Equation: transport equation for V0}
		2g_{ab}\gamma^{\alpha\beta}(\nabla^\psi V_0)^{bc}_{\phantom{bc}\alpha}(\mathrm{d}\sigma)_\beta+g_{ab}V_0^{bc}(\Delta_\gamma\sigma-2)&=0\\
		\label{Equation: transport equation for Vn}
		E(V_{n-1})^c_{\phantom{c}a}+2ng_{ab}\gamma^{\alpha\beta}(\nabla^\psi V_n)^{bc}_{\phantom{bc}\alpha}(\mathrm{d}\sigma)_\beta+ng_{ab}V_n^{bc}\big(\Delta_\gamma\sigma+2(n-1)\big)&=0\,.
		\end{align}
	\end{subequations}
	Notice that the latter is a system of transport equations which can be solved recursively once we provide initial conditions for the tensors $V_n^{bc}$.
	The customary choice for the initial data is to consider the limit $x\to x^\prime$ of equation \eqref{Equation: hierarchy of equation for Vns}.
	Denoting with $[A](x):=\lim_{x\to x^\prime}A(x,x^\prime)$ the coinciding point limit of a generic smooth bi-tensor -- \textit{cf.} remark \ref{Remark: on exterior tensor product and symmetrix sections} -- we get
	\begin{align}\label{Equation: initial data for system of transport equations for Vn}
	[E(V_0)^c_{\phantom{c}a}]+2[g_{ab}V_1^{bc}]=0\,,\qquad
	[E(V_{n-1})^c_{\phantom{c}a}]+2n^2[g_{ab}V_n^{bc}]=0\,,
	\end{align}
	where we used the identities
	\begin{align}
	[\sigma]=0\,,\qquad[(\mathrm{d}\sigma)_\alpha]=0\,,\qquad[(\nabla^\Sigma\circ\nabla^\Sigma\sigma)_{\alpha\beta}]=\gamma_{\alpha\beta}\,.
	\end{align}
	Notice that the equations in \eqref{Equation: initial data for system of transport equations for Vn} specify initial data for $V_n^{bc}$ for all $n\geq 1$, leaving us only with an arbitrariness in the choice of the initial datum for $V_0$. In order for $EP-\operatorname{Id}\,,PE-\operatorname{Id}\in\Gamma(\psi^*TM\boxtimes\psi^*T^*M)$, we fix
	\begin{align}\label{Equation: initial condition for V0}
	[V_0^{bc}]=g^{bc}\,.
	\end{align}
	
	We now consider the Hadamard parametrix $H_\lambda$ associated with $E$ and background data $(N;b_\lambda)$.
	Once again we have
	\begin{align*}
		H^{bc}_\lambda(x,x')=\sum_{n\geq 0}V_{\lambda,n}^{bc}(x,x')\sigma_\lambda(x,x')^n\log\sigma_\lambda(x,x')\,,
	\end{align*}
	where $\sigma_\lambda$ is the halved squared geodesic distance built out of the metric $\lambda^{-2}\gamma_{\alpha\beta}$.
	The smooth tensors $V_{\lambda,n}^{bc}$ satisfy the system \eqref{Equation: hierarchy of equation for Vns} with background data $(N;b_\lambda)$.
	Observe that equation \eqref{Equation: transport equation for V0} is invariant under scaling $\gamma_{\alpha\beta}\to\lambda^{-2}\gamma_{\alpha\beta}$ because $\sigma_\lambda=\lambda^{-2}\sigma$.
	Together with the initial conditions $[V_0]^{bc}=[V_{\lambda,0}]^{bc}=g^{bc}$ this entails $V_{\lambda,0}=V_0$.
	By induction it easily follows from the scaling behaviour of equations \eqref{Equation: transport equation for Vn} that $V_{\lambda,n}=\lambda^{2n}V_n$.
	Therefore
	\begin{align*}
		H_\lambda^{bc}-H^{bc}&=
		\sum_{n\geq 0}V_{\lambda,n}^{bc}\lambda^{-2n}\sigma^n\big(\log\sigma-2\log\lambda\big)-
		\sum_{n\geq 0}V_n^{bc}\sigma^n\log\sigma=
		-2\log(\lambda)V^{bc}\,.
	\end{align*}
	Using the initial condition \ref{Equation: initial condition for V0}, equation \eqref{Equation: coinciding point limit of scaled and non-scaled Hadamard parametrix} follows.
\end{proof}

\section{The Peetre-Slov\'ak theorem}\label{Appendix: Peetre-Slovak theorem}
In this section we recall succinctly the Peetre-Slov\'ak theorem as well as all ancillary definitions. For more details, we refer to \cite{Navarro-Sancho-14} and especially to \cite[Appendix A]{Khavkine:2014zsa}, to which this appendix is inspired.
In the following $E\stackrel{\pi_E}{\to} B, F\stackrel{\pi_F}{\to}B$ are smooth bundles over a smooth manifold $B$, while $J^rE$ denotes the $r$-jet bundle over $B$ for $r\in\mathbb{N}$ -- refer to \cite{Kolar-Michor-Slocvak-93} for definitions and properties.
\begin{Definition}\label{Definition: differential operator}
	A map $D\colon\Gamma(E)\to\Gamma(F)$ is a called a {\em differential operator} of globally bounded order $r\in\mathbb{N}$ if there exists a smooth map $d\colon J^rE\to F$ such that $\pi_F\circ d=\pi_{J^rE}$ and 
	\begin{align}\label{Equation: definition of differential operator of globally bounded order}
	D(\varepsilon)=d(j^r\varepsilon)\qquad\forall\varepsilon\in\Gamma(E)\,,
	\end{align}
	where $j^r\varepsilon\in\Gamma(J^rE)$ denotes the $r$-jet extension of $\varepsilon$.
\end{Definition}

\begin{Definition}
	A map $D\colon\Gamma(E)\to\Gamma(F)$ is called a {\em differential operator of locally bounded order} if for all $x_0\in B$ and for al $\varepsilon_0\in\Gamma(E)$, there exists
	\begin{enumerate}
		\item an open subset $U\subseteq B$ containing $x_0$ and with compact closure,
		\item an integer $r\in\mathbb{N}$, as well as a neighbourhood $Z^r\subseteq J^rE$ of $j^r\varepsilon_0(U)$ such that $\pi_{J^rE}Z^r=U$,
		\item a  smooth map $d\colon Z^r\to F$ such that $\pi_F\circ d=\pi_{J^rE}$
	\end{enumerate} 
   so that
	\begin{align}\label{Equation: definition of differential operator of locally bounded order}
	D(\varepsilon)(x)=d(j^r\varepsilon)(x)\,,
	\end{align}
	for all $x\in U$ and $\varepsilon\in\Gamma(E)$ with $j^r\varepsilon(U)\subseteq Z^r$.
\end{Definition}
The Peetre-Slov\'ak's Theorem gives a sufficient condition for a map $D\colon\Gamma(E)\to\Gamma(F)$ to be a differential operator of locally bounded order.

In addition recall that, denoting with $\pi_d\colon B\times\mathbb{R}^d\to B$ the canonical projection to $B$, the pull-back bundle $\pi_d^*E\stackrel{\pi_{\pi_d^*E}}{\to}B\times\mathbb{R}^d$ is the smooth bundle defined by
\begin{align}\label{Equation: pull-back bundle}
\pi^*E:=\{(s,x,e)\in\mathbb{R}^d\times B\times E|\;\pi_E(e)=\pi_d(s,x)\}\simeq\mathbb{R}^d\times E\,.
\end{align}
Denoting with $\pi_{d,E}$ the projection $\pi_{d,E}\colon\pi_d^*E\to E$, each smooth section $\zeta\in\Gamma(\pi_d^*E)$ induces a smooth family of sections $\{\zeta_s\}_{s\in\mathbb{R}^d}$ in $\Gamma(E)$ defined by $\zeta_s(x):=\pi_{d,E}\zeta((s,x))$ which, in turn, depends smoothly on the parameter $s\in\mathbb{R}^d$.
\begin{Definition}\label{Definition: smooth compactly supported d-dimensional family of variations}
	Let $d\in\mathbb{N}$ and let $\{\zeta_s\}_{s\in\mathbb{R}^d}$ be a smooth family of sections in $\Gamma(E)$ induced by a smooth section $\zeta\in\Gamma(\pi_d^*E)$.
	We say that $\{\zeta_s\}_{s\in\mathbb{R}^d}$ is a smooth compactly supported $d$-dimensional family of variations if there exists a compact $K\subseteq B$ such that $\zeta(s,x)=\zeta(s',x)$ for all $x\notin K$ and for all $s,s'\in\mathbb{R}^d$.
\end{Definition}
\begin{Definition}\label{Definition: weak regularity condition}
	A map $D\colon\Gamma(E)\to\Gamma(F)$ is called weakly-regular if, for all  $d\in\mathbb{N}$ and for all smooth compactly supported $d$-dimensional families of variations $\{\zeta_s\}_{s\in\mathbb{R}^d}$ -- see Definition \ref{Definition: smooth compactly supported d-dimensional family of variations} -- $\psi_s:=D\zeta_s$ is a smooth compactly supported $d$-dimensional family of variations.
\end{Definition}
\begin{theorem}[Peetre-Slov\'ak]\label{Theorem: Peetre-Slovak's theorem}
	Let $D\colon\Gamma(E)\to\Gamma(F)$ be a smooth map such that
	\begin{itemize}
		\item
		for all $\varepsilon\in\Gamma(E)$ and for all $x\in B$, $D\varepsilon(x)$ depends only on the germ of $\varepsilon$ at $x\in B$, i.e. $(D\varepsilon)(x)=(D\widetilde{\varepsilon})(x)$ for all $\widetilde{\varepsilon}\in\Gamma(E)$ which coincides with $\varepsilon$ in a neighbourhood of $x$;
		\item
		$D$ is weakly regular as per Definition \ref{Definition: weak regularity condition}.
	\end{itemize}
	Then $D$ is a differential operator of locally bounded order as per Definition \ref{Definition: differential operator}.
\end{theorem}

\section{Fulfilment of the perturbative agreement}

In this section we comment on the principle of perturbative agreement (PPA for short) for the model we have introduced in Definition \ref{Definition: locally covariant theory of interest}.

The PPA has been introduced in \cite{Hollands-Wald-05} as a further constraint on the structure of Wick powers -- see also \cite{Drago-Hack-Pinamonti-2016,Zahn-13}.
Loosely speaking, it requires that a theory associated with a quadratic perturbation $E_s$ of the elliptic operator $E$ introduced in equation \eqref{Equation: definition of the elliptic operator} should yield to an algebra $\mathcal{A}_s[N;b]$ compatible with the unperturbed algebra $\mathcal{A}[N;b]$.
Here $E_s-E\in\Gamma_{\mathrm{c}}(\mathrm{S}^{\otimes 2}\psi^*T^*M)$ is a smooth and compactly supported ($1$-dimensional) family of variation.
The compatibility between $\mathcal{A}_s$ and $\mathcal{A}$ is in the sense of formal power series in $s$ -- \textit{cf.} definition \ref{Definition: PPA}.

As pointed out in \cite{Zahn-13} the PPA is important in our setting because, among other things, it ensures that the renormalization group flow technique we applied in Section \ref{Section: Renormalization and Ricci flow} does not depend on the splitting $\mathcal{L}=\mathcal{L}_{\mathrm{free}}+\mathcal{L}_{\mathrm{int}}$.
A complete discussion of the PPA is out not within the scopes of this paper -- for a complete discussion in the Riemannian setting see \cite{Dappiaggi-Drago-Rinaldi-19}.
In the present appendix we provide a brief resum\'e of the content of the PPA, proving that there exists a family of Wick powers as per definition \ref{Definition: Wick power associated with a locally covariant observable} which fulfils it -- \textit{cf.} Proposition \ref{Proposition: fulfilment of the PPA}.

In what follows $E_s$ will always denote a smooth and compactly supported ($1$-dimensional) family of variation -- \textit{cf.} Definition \ref{Definition: smooth compactly supported d-dimensional family of variations} -- of the elliptic operator $E$ defined as per equation \eqref{Equation: definition of the elliptic operator}. In particular $E_s$ is elliptic for all $s$.
Notice that, for the sake of simplicity, we are assuming that $E_s-E\in\Gamma_{\mathrm{c}}(\mathrm{S}^{\otimes 2}\psi^*T^*M)$ is a differential operator of order at most $1$.
This is actually enough for our setting see however \cite{Dappiaggi-Drago-Rinaldi-19,Hollands-Wald-05} for completeness.

\paragraph{Formulation of the PPA.}
In order to formulate the PPA a few preliminary definitions are in due order.
First of all we need a linear isomorphism $R_s\colon\operatorname{Par}[N;b]\ni P\to P_s\in\operatorname{Par}_s[N;b]$ between the space of parametrices $\operatorname{Par}_s[N;b]$ associated with $E_s$ and those of $E$.
The construction of this map is rather standard, see \cite{Dappiaggi-Drago-Rinaldi-19} for further details and \cite{Dappiaggi-Drago-16,Drago-Hack-Pinamonti-2016,Hollands-Wald-05,Zahn-13} for the corresponding map in the Lorentzian setting.
For what concerns the PPA, we just need the perturbative expansion of $R_s$ as a formal power series in $s$ up to a smooth remainder. Let $P\in\operatorname{Par}[N;b]$; since $Q_s:=E_s-E$ is compactly supported,
\begin{align*}
	E_s=
	E+Q_s=
	E(\operatorname{Id}+PQ_s)-SQ_s\,,
\end{align*}
where $S\in\Gamma(\psi^*TM\boxtimes\psi^*T^*M)$ is such that $PE-\operatorname{Id}_{\Gamma_{\mathrm{c}}(\psi^*T^*M)}=S$ -- \textit{cf.} equation \eqref{Equation: defining property of parametrix}.
We consider the map $R_{[[s]]}\colon\Gamma_{\mathrm{c}}(\psi^*T^*M)\to\Gamma(\psi^*TM)[[s]]$ defined by
\begin{align}\label{Equation: definition of classical moeller map between parametrices}
	R_{[[s]]}\omega:=\sum_{n\geq 0}(-PQ_s)^nP\omega\,.
\end{align}
This map can be interpreted as a perturbative expansion (up to a smooth remainder) of a well-defined isomorphism $R_s\colon\operatorname{Par}[N;b]\to\operatorname{Par}_s[N;b]$ -- \textit{cf.} \cite{Dappiaggi-Drago-16,Dappiaggi-Drago-Rinaldi-19,Drago-Hack-Pinamonti-2016,Hollands-Wald-05}.

The second ingredient we need is a $*$-isomorphism $\beta_s\colon\mathcal{A}_{\textrm{reg}}[N;b]\to\mathcal{A}_{s,\textrm{reg}}[N;b]$.
Here the subscript $_{\textrm{reg}}$ denotes the algebra generated by regular local functional, namely those with smooth functional derivatives of all orders.
We will not enter into the details of this construction, however, we give the explicit form for $\beta_s$:
\begin{align}\label{Equation: beta map}
	(\beta_s F)(P_s):=\exp\big[\Upsilon_{P_s-P}\big]F[P]\qquad\forall F\in\mathcal{A}[N;b]\,.
\end{align}
This can be extended to a map $\beta_{[[s]]}\colon\mathcal{A}[N;b]\to\Gamma(\mathcal{E}[N;b])[[s]]$ with values in the algebra of formal power series in $s$ with coefficients in $\Gamma(\mathcal{E}[N;b])$ -- \textit{cf.} Definitions \ref{Definition: bundle of fiber algebras}-\ref{Definition: locally covariant theory of interest}.
The expansion is possible since, on account of equation \eqref{Equation: definition of classical moeller map between parametrices}, $P_s-P=P_{[[s]]}-P+R=\sum_{n\geq 1}(-PG_s)^nP+R$ has a well-defined coinciding point limit (here $R$ is a smooth remainder). Therefore $\beta_{s}$ is well-defined at each order in $s$.
As explained in \cite{Dappiaggi-Drago-Rinaldi-19,Drago-Hack-Pinamonti-2016} the map $\beta_{[[s]]}$ can be interpreted as an extension of the expansion in formal power series of $\beta_s$.

We focus on Wick powers, strengthening the smoothness requirement of Definition \ref{Definition: Wick power associated with a locally covariant observable} by allowing also variations of the elliptic operator $E_s$.
\begin{Definition}\label{Definition: strengthening of smoothness condition for Wick powers}
	Let $d,n\in\mathbb{N}$ and let $(N;b_s)\in\operatorname{Obj}(\mathsf{Bkg})$ be such that $\{b_s=(\psi,\gamma_s,g_s)\}_{s\in\mathbb{R}^d}$ is a smooth, compactly supported $d$-dimensional family of variations of $b=(\psi,\gamma,g)$ as per Definition \ref{Definition: smooth compactly supported d-dimensional family of variations}.
	Moreover let $E_{t,s}$ be a smooth and compactly supported $n$-dimensional family of variations of the elliptic operator $E_s$ constructed out of the background data $b_s$ as per equation \eqref{Equation: definition of the elliptic operator}.
	For all smooth families $\{P_{t,s}\}_{s\in\mathbb{R}^d}$ where $P_{t,s}\in\operatorname{Par}_t(N,b_s)$ is a parametrix for $E_{t,s}$ for all $s\in\mathbb{R}^d$ and $t\in\mathbb{N}^n$, let $\mathcal{U}_k\in\Gamma_{\mathrm{c}}(\pi_{d+n}^*\mathrm{S}^{\otimes k}\psi^*T^*M)'$ be the distribution defined by
	\begin{align}
		\mathcal{U}_k(\chi\otimes\omega_1):=
		\int_{\mathbb{R}^{d+n}}\mathrm{d}s\mathrm{d}t\,\Phi^k_t[N;b_s](\omega_1,P_{t,s},0)\chi(s,t)\,,
	\end{align}
	where $\omega_1\in\mathrm{S}\Gamma_{\mathrm{c}}^{k,1}[N;b],\chi\in C^\infty_{\mathrm{c}}(\mathbb{R}^{d+n})$.
	Here $\Phi^k_t[N;b_s]$ denotes the ($k$-th) Wick power associated with the background data $(N;b_s)$ and with the elliptic operator $E_{s,t}$. If $\textrm{WF}(\mathcal{U}_k)=\emptyset$, we call the family of Wick powers $\Phi^\bullet$ smooth.
\end{Definition}

\begin{remark}
	Loosely speaking Definition \ref{Definition: strengthening of smoothness condition for Wick powers} requires a suitable smoothness of $\Phi^\bullet$ with respect both to the background data $b$ and to the variation of the elliptic operator $E$.
	For certain models -- like the scalar field \textit{cf.} \cite{Dappiaggi-Drago-Rinaldi-19} -- the variations of the background data exhaust all possible variations of the associated elliptic operator $E$.
	In this situation the smoothness as per Definition \ref{Definition: strengthening of smoothness condition for Wick powers} coincides with the one required in Definition \ref{Definition: Wick power associated with a locally covariant observable}.
\end{remark}

\begin{remark}
	A smooth family of parametrices $P_{t,s}\in\operatorname{Par}_t[N;b_s]$ can be constructed by setting $P_{t,s}:=R_{t}P_s$ where $P_s\in\operatorname{Par}[N;b_s]$ is a smooth family of parametrices for $E_s$ -- \textit{cf.} Remark \ref{Remark: on the regularity property}.
\end{remark}

From now on $\Phi^\bullet$ will denote a family of Wick powers as per Definition \ref{Definition: Wick power associated with a locally covariant observable} satisfying the smoothness requirement of Definition \ref{Definition: strengthening of smoothness condition for Wick powers}.
Notice that the family $\wick{\Phi^\bullet}$ defined in Example \ref{Equation: definition of Wick powers exploited for Ricci flow} satisfies such smoothness requirement.

\begin{Definition}\label{Definition: PPA}
	Let $E_s$ denote a smooth and compactly supported ($1$-dimensional) family of variation -- \textit{cf.} definition \ref{Definition: smooth compactly supported d-dimensional family of variations} -- of the elliptic operator $E$ defined as per equation \ref{Equation: definition of the elliptic operator}.
	We say that the family of Wick powers $\Phi^\bullet$ satisfies the principle of perturbative agreement (PPA) if for all $k\geq 2$, $n\in\mathbb{N}\cup\lbrace 0\rbrace$, $P\in\operatorname{Par}[N;b]$, $\omega_m\in\mathrm{S}\Gamma_{\mathrm{c}}^{k,m}[N;b]$ it holds
	\begin{align}\label{Equation: PPA}
		\frac{\mathrm{d}^n}{\mathrm{d}s^n}\Phi^k_s[N;b](\omega_m,P_s)\bigg|_{s=0}=
		\frac{\mathrm{d}^n}{\mathrm{d}s^n}\beta_{s}(\Phi^k[N;b])(\omega_m,P_{s})\bigg|_{s=0}\,.
	\end{align}
\end{Definition}

\begin{remark}
	A direct computation shows that the PPA is satisfied if and only if equation \eqref{Equation: PPA} holds true for $n=1$ -- \textit{cf.} \cite{Dappiaggi-Drago-Rinaldi-19}.
	Moreover, on account of the lack of renormalization ambiguities for $m\geq 2$ -- \textit{cf.} Remarks \ref{Remark: on the log-singularity of the parametrices}-\ref{Remark: on the general case of dimension D greater than two} -- the PPA is fulfilled whenever it holds for $m=1$.
\end{remark}

\noindent We state the main result of this appendix.
\begin{proposition}\label{Proposition: fulfilment of the PPA}
	The family $\wick{\Phi^\bullet}$ defined in equation \eqref{Example: Wick powers exploited for Ricci flow} satisfies the PPA as per Definition \ref{Definition: PPA} with respect to a family of variations $E_s$ of the elliptic operator $E$ such that $E_s-E\in\Gamma_{\mathrm{c}}(\mathrm{S}^{\otimes 2}\psi^*T^*M)$ is a differential operator of order at most $1$.
\end{proposition}
\begin{proof}
	Observe that on account of Theorems \ref{Theorem: ambiguities for Wick powers}-\ref{Theorem: structural form of ambiguities} we may write for all $k\geq 2$ and $\omega\in\mathrm{S}\Gamma_{\mathrm{c}}^{k,1}[N;b]$
	\begin{align*}
		\Phi^k_s[N;b](\omega_1)=\;
		\wick{\Phi^k}_s[N;b](\omega_1)+
		\sum_{\ell=0}^{k-2}\wick{\Phi^\ell}_s\big(c_{s,k-\ell}[N;b]\lrcorner\omega_1\big)\,,
	\end{align*}
	where $\wick{\Phi^k}$ is defined as in Example \ref{Example: Wick powers exploited for Ricci flow} while $c_{s,\ell}\in\Gamma^{\ell,1}[N;b]$ satisfies the hypothesis of theorem \ref{Theorem: ambiguities for Wick powers}-\ref{Theorem: structural form of ambiguities}.
	Moreover, $c_{s,\ell}$ is a smooth and compactly supported family of variation and we set $c_\ell:=c_{s,\ell}|_{s=0}$.
	
	Our aim is to show that $\wick{\Phi^\bullet}$ satisfies equation \eqref{Equation: PPA}.
	In particular we shall impose equation \eqref{Equation: PPA} for a generic family $\Phi^\bullet$ of Wick powers.
	This will constraint the coefficient $c_{s,\ell}$ defined above, in particular we shall prove that equation \eqref{Equation: PPA} implies that we can choose $c_{s,\ell}=0$, that is, $\Phi^\bullet=\wick{\Phi^\bullet}$.
	We first consider the case $k=2$.
	Setting $\delta:=\frac{\mathrm{d}}{\mathrm{d}s}\big|_{s=0}$, by direct inspection it holds
	\begin{align*}
		\delta[\Phi^2[N;b](\omega_1,P)]&=
		\delta\big[\wick{\Phi^2}[N;b](\omega_1,P)+C_{s,2}[N;b](\omega_1)\big]\\&=
		\Upsilon_{\delta(W_P)}\wick{\Phi^2}[N;b](\omega_1,P)+\delta\big[C_2[N;b](\omega_1))\big]\,,
	\end{align*}
	where $C_2[N;b](\omega_1):=\int_\Sigma\langle c_2[N;b],\omega_1\rangle$.
	Similarly the first order in $s$ in the right hand side of equation \eqref{Equation: PPA} reads
	\begin{align*}
		\frac{\mathrm{d}}{\mathrm{d}s}\beta_s\Phi^k[N;b](\omega_1,P_s)\bigg|_{s=0}=
		\Upsilon_{\delta(P)}\Phi^2[N;b](\omega_1,P)=
		\Upsilon_{\delta(P)}\wick{\Phi^2}[N;b](\omega_1,P)\,,
	\end{align*}
	where we exploited equations \eqref{Equation: definition of classical moeller map between parametrices}-\eqref{Equation: beta map}.
	Equation \eqref{Equation: PPA} entails
	\begin{align*}
		\delta\big[C_2[N;b](\omega_1))\big]=
		-\Upsilon_{\delta(W_P)+\delta(P)}\wick{\Phi^2}[N;b](\omega_1,P)=
		\langle[\delta(H)],\omega_1\rangle\,,
	\end{align*}
	where we used equation \eqref{Equation: inductive condition on derivative for Wick powers} and Remark \ref{Remark: local Hadamard representation of parametrices}.
	Therefore the PPA for $k=2$ can be fulfilled if
	\begin{align}\label{Equation: PPA condition for second coefficient}
		\delta(c_2)=-[\delta(H)]\,.
	\end{align}
	Since $\delta(H)$ is local and covariant the above equation can be considered as a definition of the coefficient $c_2$ -- indeed, it respects all requirement of Theorems \ref{Theorem: ambiguities for Wick powers}-\ref{Theorem: structural form of ambiguities}.
	In particular $c_{2,s}$ is the solution to the ODE $\delta(c_{2,s})=-[\delta(H_s)]$.
	
	For the general case we compute once again the right and the left hand side of equation \eqref{Equation: PPA} using equation \eqref{Equation: ambiguities for Wick powers}. The final result is
	\begin{align*}
		0&=\sum_{\ell=0}^{k-2}\wick{\Phi^\ell}[N;b](\delta(c_{k-\ell}[N;b])\lrcorner\omega_1,P)\\&-
		\Upsilon_{[\delta(H)]}\bigg[\wick{\Phi^k}[N;b](\omega_1,P)+\sum_{\ell=0}^{k-2}\wick{\Phi^\ell}[N;b](c_{k-\ell}[N;b]\lrcorner\omega_1,P)\bigg]\,,
	\end{align*}
	Equation \eqref{Equation: inductive condition on derivative for Wick powers} leads to the following generalization of equation \eqref{Equation: PPA condition for second coefficient}
	\begin{align*}
		\delta(c_{k-\ell})=-(\ell+2)(\ell+1)[\delta(H)][\otimes]c_{k-2-\ell}\qquad 2\leq\ell\leq k-3\,.
	\end{align*}
	Once again this can be used to define inductively the coefficients $c_{\ell}$.
	
	We now prove that $[\delta(H)]=0$, which implies that the choice $c_{s,\ell}=0$ leads to a family of Wick powers satisfying the PPA.
	We recall that we are considering families of variations $E_s$ such that $E_s-E\in\Gamma_{\mathrm{c}}(\mathrm{S}^{\otimes 2}\psi^*T^*M)$ is a differential operator of order at most $1$.
	For $\varphi\in\Gamma(\psi^*TM)$ we can write locally
	\begin{align}
		(E_s\varphi-E\varphi)_a=(A_s)^{\alpha}_{\phantom{\alpha}ab}(\nabla^\psi\varphi)^b_\alpha+(T_s)_{ab}\varphi^b\,,
	\end{align}
	where $(A_s)^{\alpha}_{\phantom{\alpha}ab},(T_s)_{ab}$ are suitable smooth tensors.
	This implies that the Hadamard parametrix $H_s$ associated with $E_s$ has the form
	\begin{align*}
		H_s=\sum_{n\geq 0}V_{s,n}\sigma^n\log\sigma\,,
	\end{align*}
	where $\sigma$ does not depend on $s$ since so it does the principal symbol of $E_s$ \cite{Garabedian-98}.
	The tensors $V_{s,n}\in\mathrm{S}\Gamma(\psi^*TM^{\boxtimes 2})$ satisfies a hierarchy of transport equations analogous to system \eqref{Equation: hierarchy of equation for Vns}.
	In particular $V_{s,0}$ satisfies
	\begin{align}\label{Equation: transport equation for Vs0}
		2g_{ab}\gamma^{\alpha\beta}(\nabla^\psi V_{s,0})^{bc}_{\phantom{bc}\alpha}(\mathrm{d}\sigma)_\beta+
		g_{ab}V_{s,0}^{bc}(\Delta_\gamma\sigma-2)+
		(A_s)^\alpha_{\phantom{\alpha}ab}V^{bc}_{s,0}(\mathrm{d}\sigma)_\alpha=0\qquad[V_{s,0}]^{bc}=g^{bc}\,.
	\end{align}
	It then follows that $\delta(V_0)$ satisfies a transport equation with initial condition $[\delta(V_0)]=0$.
	This implies that $\delta(V_0)=0$ and therefore
	\begin{align}
		\delta(H)=\sum_{n\geq 1}\delta(V_n)\sigma^n\log\sigma=O(\sigma)\,.
	\end{align}
\end{proof}

\section*{Acknowledgments}
The work of C.~D.\ was supported by the University of Pavia, while that of N.~D.\ was supported in part by a research fellowship of the University of Pavia.
We are grateful to Federico Faldino, Igor Khavkine, Alexander Schenkel and Jochen Zahn for the useful discussions.
We are especially grateful to Klaus Fredenhagen for the enlightening discussions on the r\^ole of the algebra of functionals.
This work is based partly on the MSc thesis of P.~R. .

\end{document}